\documentclass[11pt, letter]{article}
\usepackage{amsthm}
\usepackage[margin =1in]{geometry}
\usepackage{palatino}
\usepackage{xcolor}
\usepackage[flushleft]{threeparttable}
\usepackage[width=\textwidth]{caption}
\usepackage{subcaption}
\usepackage[pdftex]{graphicx}
\usepackage[skip=0pt]{caption}
\usepackage{setspace}
\usepackage{multirow}
\usepackage{multibib}
\usepackage[multiple]{footmisc}
\usepackage[round]{natbib}
\usepackage{verbatim}
\usepackage[T1]{fontenc}
\usepackage[utf8x]{inputenc}
\usepackage[pdftex]{graphicx}
\usepackage{hologo}
\usepackage{textcomp}
\usepackage{hanging}
\usepackage{amsmath}
\usepackage{amssymb}
\usepackage{amsfonts}
\usepackage{color}
\usepackage{tikz}
\usepackage{float}
\usepackage{epigraph}
\usepackage{authblk}
\usepackage{tabularx}
\usepackage{xr-hyper}
\usepackage{xurl}
\usepackage[hidelinks]{hyperref}
\usepackage{cleveref}

\newtheorem{theorem}{Theorem}
\newtheorem{lemma}{Lemma}
\newtheorem{proposition}{Proposition}

\newtheorem{assumption}{Assumption}
\allowdisplaybreaks

\begin{document}

\title{\Large AI Information Synthesis and Content Polarization}

    \author{Chang Liu\thanks{Lecturer in Economics, UNSW Business School, chang.liu36@unsw.edu.au.}\hspace{2cm}
Yuting Zhu\thanks{Assistant Professor of Marketing, National University of Singapore, NUS Business School, y.zhu@nus.edu.sg.} }
\footnotetext[1]{The authors are listed alphabetically and contributed equally to this work. The authors thank Benjamin Golub, Shengwu Li, Eric Maskin, and Tomasz Strzalecki, as well as seminar participants at Hong Kong Baptist University and audiences at the 2026 CCF Annual Conference on Computational Economics, for helpful comments.}

\date{\today}

\maketitle

\begin{abstract}
\singlespacing
\noindent
Artificial intelligence (AI) can synthesize dispersed information into balanced and informative content, raising the possibility that AI may help reduce polarization. We show that this promise can be undermined by the strategic response of content producers. We develop a game-theoretic model in which consumers acquire information to make downstream decisions, content producers allocate limited informational capacity between broadly valuable vertical information and horizontally differentiated information, and AI provides an alternative source of synthesized information. The effects of AI unfold in three phases as its reliability increases. When reliability is low, AI has no effect on the market. At intermediate levels, AI disciplines prices without attracting consumers or changing content. When AI becomes sufficiently reliable, consumers begin to substitute toward it, and content producers respond by shifting informational capacity away from the vertical dimension and toward horizontally differentiated content, thereby increasing content polarization. The mechanism operates through audience composition: AI attracts consumers who are less well matched to individual producers, leaving producers with stronger incentives to cater to their remaining audiences. The polarization effect intensifies with AI reliability but weakens with greater competition among content producers. Our findings highlight a potential wedge between AI's direct and equilibrium effects: a technology that provides consumers with more balanced information can indirectly make the broader content ecosystem more polarized.

\singlespacing
\noindent\emph{\textbf{Keywords}}: Generative AI, Polarization, Content Markets, Media Economics
\end{abstract}

\newpage

\doublespacing

\section{Introduction} \label{sec:intro}
Polarization has become one of the central challenges in today's information environment, and content markets sit at the center of this concern because they shape what information is produced, how it is framed, and how consumers consume it. Economic and marketing analysis suggests that polarized content is not simply an accidental byproduct of digital technology, but can arise as an endogenous equilibrium outcome of market forces (e.g., \citealt{marketfornews, xiangsarvary2007, gentzkowshapiro2010, peregoyukesel2022}). Mitigating polarization has accordingly become a shared concern across disciplines.

Recent advances in artificial intelligence (AI) have prompted computer scientists to propose a new class of solutions for reducing polarization. Central to these proposals is AI’s information-synthesis capability: large language models (LLMs) can process information and viewpoints dispersed across many sources and distill them into a coherent, potentially more balanced representation \citep{ragsynthesis}. This capability can be delivered through a variety of product formats, including standalone chatbots, conversational search, and AI-generated summaries. Google DeepMind’s ``Habermas Machine,'' for example, demonstrates the potential of AI synthesis by using an LLM mediator to synthesize opposing opinions into consensus statements \citep{googlemachine}. Importantly, AI synthesis is increasingly becoming part of consumers’ information choices and competing with content obtained directly from original content providers. \citet{reuters2026dnr} reports that 10\% of consumers globally now use standalone AI chatbots such as ChatGPT, Perplexity, and Gemini for news, up from 7\% in the previous year. Consistent with such substitution, \citet{agarwalsen2026} show in a field experiment that Google AI Overviews reduce outbound organic clicks by 39.8\% and increase zero-click searches by 34.5\%, while \citet{padillaetal2026} find that LLM adoption leads to substantial declines in traditional search and traffic to some downstream content providers. Taken together, these developments suggest that AI synthesis is emerging as an alternative way for consumers to acquire information rather than merely as a tool layered onto existing content consumption. This raises a natural possibility: by providing consumers with synthesized and potentially more balanced information, AI may mitigate the polarizing dynamics of content markets.

This optimism, however, takes a partial-equilibrium view: it emphasizes the balanced information that AI synthesis can provide directly to consumers while holding fixed the strategic choices of content producers.\footnote{In this paper, we abstract from the institutional identity of content producers, whether traditional news publishers, social media creators, independent commentators, or other information providers, in order to focus on the core economic mechanism of interest.} In reality, content producers are strategic: they choose what content to produce and how to position it in response to the competitive environment. The availability of AI-synthesized information as a competing alternative changes this environment and, in turn, reshapes content producers' equilibrium incentives. A natural intuition is that a balanced, high-quality AI alternative should discipline content producers: by offering consumers an attractive source of information, AI may pressure producers to improve their content or adopt less polarized positions in order to remain competitive. Yet the strategic response may not follow this intuition. Faced with stronger competition from AI, content producers may instead have greater incentives to differentiate their offerings, potentially shifting toward more polarized content. The equilibrium effect of AI synthesis can therefore be striking: the same technology that directly provides consumers with more balanced information may indirectly induce strategic content producers to supply more polarized content.\footnote{\citet{zhaoberman2026} also provide empirical evidence that publishers are adjusting the supply of content in response to generative AI along multiple margins, including restricting LLM access and shifting toward richer and more interactive content formats. Our analysis complements this evidence by developing a theoretical framework that endogenizes content producers' editorial choices and characterizes how these choices respond to competition from AI-mediated information.}

In this paper, we examine how strategic content producers respond to the availability of AI-synthesized information and ask whether their equilibrium supply-side response reinforces or undermines AI's promise as a depolarizing technology. We address this question using a game-theoretic model of content competition with microfounded consumer demand for information.\footnote{AI information synthesis is an emerging technology whose institutional form, adoption patterns, and monetization rules remain unsettled. If researchers wait until its polarization effects are fully observable, content strategies and consumption habits may already have become entrenched. Game theory allows us to enter this debate ex ante.} Consumers seek information because they are uncertain about a payoff-relevant event and want to make better downstream decisions. Content producers strategically choose what information to provide by allocating limited informational capacity across different dimensions of the event. One dimension is vertical, capturing information that all consumers value in the same direction, such as a product's objective quality; the remaining dimensions are horizontal, capturing aspects over which consumers differ in their preferences, such as product design or lifestyle fit. This structure endogenizes a key tradeoff in content production: devoting more informational capacity to broadly valuable information necessarily leaves less capacity for horizontally differentiated information, and vice versa. We model AI as an alternative information source whose commonly known reliability captures the instrumental value of AI synthesis to consumers, using a truth-or-noise information structure \citep{lewissappington1994,johnsonmyatt2006}. Consumers choose between AI-synthesized information and strategically produced content, allowing us to study how the availability of AI changes producers' equilibrium incentives over the type of information they supply.

Our main results show that the effects of AI synthesis unfold in three phases as its reliability improves. When AI reliability is low, AI is irrelevant to the market: content producers make the same editorial choices and charge the same prices as they would in its absence. At intermediate levels of reliability, AI disciplines prices without attracting any consumers in equilibrium. Notably, the first beneficiaries are consumers who are best matched to each producer along the horizontal dimension, as the competitive threat from AI forces producers to lower the prices charged to these consumers. Only when AI becomes sufficiently reliable does it begin to draw consumers away from content producers. At this point, AI competition reduces producers' incentives to emphasize the vertical information dimension that is broadly valued across consumers, causing their content to become more horizontally specialized and, in this sense, more polarized. This editorial distortion becomes stronger as AI reliability increases and is more pronounced when the market contains relatively few content producers. By contrast, markets in which consumers place greater weight on the vertical dimension are more insulated from this effect, as AI must reach a higher level of reliability before producers' editorial incentives are distorted.

The mechanism behind this polarization result operates through a change in the audience that content producers compete to serve. In the absence of sufficiently strong AI competition, producers have an incentive to provide broadly valuable vertical information because doing so helps them appeal to consumers beyond those who are most closely aligned with their horizontal orientation. Once AI becomes sufficiently informative, however, it provides these less well-matched consumers with an attractive alternative, allowing producers to retreat toward the consumers who are more naturally aligned with their content. This changes the relative returns to the two dimensions of content production: producers have less incentive to use their limited informational capacity to appeal broadly and greater incentive to tailor their content toward their remaining audience. AI synthesis therefore induces content polarization not by pushing producers to new or more extreme horizontal positions, but by changing the type of information they optimally provide at their existing positions. In this sense, the same technology that directly offers consumers balanced information can indirectly make the remaining producer-generated content more polarized.

We extend the model in three directions. First, we move beyond content production to content consumption and show that, once AI becomes sufficiently reliable to attract consumers, the greater horizontal specialization of producer-generated content is accompanied by greater divergence in consumers' evaluations. Second, we allow AI to personalize its output to individual consumers. The main equilibrium forces remain unchanged: personalization alters the effective informational value of AI and hence the reliability threshold at which it begins to affect the market, but preserves the strategic mechanism through which AI competition can induce more polarized content. Third, we allow AI reliability to depend endogenously on the information supplied by content producers. Our polarization result continues to hold, but a new strategic force emerges: because providing more broadly valuable information also improves the information available to AI, content producers may reduce their vertical emphasis partly to weaken the competing AI alternative. This force can induce greater polarization even before consumers actually switch to AI, as long as AI is sufficiently valuable to constrain producers' pricing.

Our analysis contributes a supply-side perspective to the emerging debate over AI and polarization. While existing discussions naturally emphasize the information that AI provides directly to consumers, we show that AI can also reshape the incentives of the content producers with which it competes, creating a potential wedge between its direct and equilibrium effects: the same technology that provides consumers with more balanced information may induce the remaining producer-generated content to become more polarized. This insight suggests that evaluations of AI synthesis should look beyond the quality, accuracy, and balance of AI output itself and account for its effects on the broader information ecosystem. 

The remainder of the paper is organized as follows. Section \ref{sec:literature} positions the paper within the related literature. Section \ref{sec:model} introduces the model setup, and Section \ref{sec:analysis} characterizes the equilibrium results. Section \ref{sec:extensions} presents relevant extensions, and Section \ref{sec:conclusion} concludes.

\section{Literature Review} \label{sec:literature}
To the best of our knowledge, this paper is the first to examine the link between AI's information-synthesis capability and the polarization of content production. Pre-AI aggregation technologies are well studied in the literature, but their primary role is to organize access to producer-generated content rather than to synthesize information across sources into a new informational product. \citet{mediaaggregators2013} model aggregators as sites that produce no original content and compete for traffic by hyperlinking to the best producer content, with the aggregator's value arising from filtering and redirecting attention. Similarly, \citet{jeonnasr2016} model an aggregator as a homepage that directs consumers, issue by issue, to the highest-quality producer article. \citet{decornieresarvary2023} study a social platform that competes for consumers' limited attention by bundling publisher content with social content, whereas \citet{amaldossdu2023} model an aggregator that displays snippets and unbundles publishers' article packages. In these models, aggregation primarily changes how consumers access and navigate existing content, while the underlying units of information remain producer-authored articles. AI introduces a qualitatively different capability: information synthesis, whereby an AI synthesizer can potentially deliver greater information value than any individual source alone. It can therefore satisfy consumers' information needs without requiring them to consume the underlying sources separately, making AI not merely an intermediary that routes consumers within the content market but a competing information-acquisition channel in its own right.\footnote{Existing AI products, such as search AI overviews and conversational chatbots, do retain links and referrals; our analysis abstracts from this familiar routing role to isolate what is economically new about AI and to understand its equilibrium implications for the content market.}

In terms of our theoretical contribution to the marketing literature, we highlight two features of our model. First, building on \citet{peregoyukesel2022}, we microfound consumers' demand for content through their information-acquisition decisions. Rather than positing an exogenous taste for content, consumers choose among information sources according to the instrumental value of the information each provides. This lets us isolate how the incremental value of information, jointly determined by content's vertical emphasis and horizontal match, disciplines equilibrium prices and editorial strategies in the presence of AI. The closest marketing antecedents are \citet{brancosunmiguel2012,brancosunmiguel2016}, who also derive the demand for information from its decision value but study a monopolist's pricing and information provision; we instead study competing content producers that jointly choose prices and editorial strategies against AI. \citet{sarvaryparker1997} similarly link the value of information to price competition among sellers, microfounding how consumers combine signals of given precision; we microfound the valuation of information itself and endogenize each content's informational composition.

A second feature of our model is that it endogenizes the tradeoff between vertical emphasis and horizontal positioning in content design. Our approach is related to \citet{andersonrenault06} and \citet{mayzlinshin11}, which model communication as subject to a limited informational bandwidth. Similarly, we assume that each piece of content has a limited capacity to convey information, reflecting both production constraints and consumers' limited attention. In this respect, the closest paper to ours is \citet{anthonydukes15}, who also endogenize a tradeoff between vertical emphasis and content that stresses horizontal match. We differ in separating vertical emphasis from horizontal positioning: in our model, these are distinct choices, so a content producer that reduces its emphasis on quality need not adjust where it positions itself horizontally.

Our paper also contributes to the growing theoretical literature on the economics of AI. Existing work has represented AI through its ability to personalize recommendations \citep{caoetal2024}, improve prediction and matching \citep{choietal2024}, or enhance content-production quality \citep{zouetal2026}. More recent work has begun to examine the informational functions of generative AI: \citet{dreyfusshoong2026} represent AI as an informative signal that assists human decision making, and \citet{keetal2026} characterize generative AI as a statistical inference engine. We follow this literature's reduced-form modeling tradition and focus on a particular economic role of AI synthesis: AI-synthesized information serves as an alternative source of information for consumers, competing with content supplied by strategic content producers.

Finally, our model speaks to the literature on content polarization. Prior research has operationalized polarization differently depending on the object of interest. \citet{xiangsarvary2007} focus on media bias, measured by the expected deviation of reported messages from the underlying truth, whereas \citet{bermankatona2020} study consumption polarization, measured by the ideological distance between consumers and the content they consume. Our notion is closer to the supply-side polarization studied in \citet{galorgeylaniyildirim2012}, \citet{yildrimgalorgeylani2013}, \citet{bondirafieianyao2026}, where polarization is associated with media outlets adopting more differentiated or extreme ideological positions. Beyond such positional differentiation, our model captures an additional margin of content polarization: producers can shift editorial emphasis away from the common vertical dimension and toward more specialized information tailored to their target audiences.

\section{The Model} \label{sec:model}
We develop a parsimonious model of content production and consumption in which consumers use publicly available information to learn about the underlying state of an event and subsequently make downstream decisions. Such an event may include, for example, the quality and features of a newly released product, the implications of an emerging health trend, or the credibility of a political development. We begin by characterizing the uncertain state about which consumers seek to form beliefs and around which content is produced. We then introduce consumer preferences and describe the strategic choices of content producers, who decide how to produce and position their content. Finally, we incorporate the role of AI and introduce the game sequence. 

\subsection{State of an Event}
The first modeling choice concerns how we represent the underlying state of an event. We model the state of an event as a three-dimensional vector $\omega=\left(\omega_{0}, \omega_{1}, \omega_{2}\right) \in \mathbb{R}^{3}$, whose three components are independently distributed as standard normals. Each component captures a distinct aspect of the event that is potentially relevant to consumers. The first component $\omega_{0}$ represents a vertical dimension that all consumers value in the same direction. For example, in a product context, this might correspond to objective quality attributes that any consumer would care about: durability, safety, or core functionality. In contrast, $\omega_1$ and $\omega_2$ represent horizontal dimensions along which consumers differ in their preferences. Using the same product example, these might correspond to aesthetic design choices, lifestyle alignment, or values-based attributes such as sustainability or brand identity, which resonate differently depending on who the consumer is.

The inclusion of the second horizontal component, $\omega_2$, allows the model to capture a richer form of content positioning. With only one horizontal component, content differentiation would collapse to a line: content producers could differ only in the extent to which they emphasize the vertical dimension relative to the horizontal dimension. By contrast, with two horizontal components, content producers can differ not merely in how much they specialize, but in \textit{which aspects} of the event they choose to emphasize.

To make this distinction concrete, consider content producers covering the launch of a new plant-based food product. One producer may focus on the product's nutritional profile and food-safety certification, information that is broadly relevant to virtually all consumers evaluating the product, corresponding to $\omega_0$. Another producer may instead emphasize the product's environmental sourcing and its fit with a sustainable lifestyle, corresponding to $\omega_1$, while a third may highlight its culinary versatility and flavor innovation for food-enthusiast audiences, corresponding to $\omega_2$. All three content producers cover the same product launch, that is, the same underlying state $\omega$, but they illuminate different aspects of that state. A consumer who primarily cares about sustainability learns substantially more from the second producer than from the first or third, while a consumer who primarily cares about culinary use cases may exhibit the opposite pattern. This form of content differentiation cannot be captured within a single horizontal dimension framework.

\subsection{Consumers}
There is a unit mass of consumers, and each consumer $i$ is characterized by a type vector $\theta_i \in \mathbb{R}^3$ that determines how she values the three components of the event state $\omega = (\omega_0, \omega_1, \omega_2)$. We adopt a linear specification in which the utility consumer $i$ derives from the event state $\omega$ is (\citealt{andersonPalmaThisse1992}):
\begin{equation}\label{eq:consumerutility}
    u(\omega, \theta_i) = \theta_{i,0}\omega_0 + \theta_{i,1}\omega_1 + \theta_{i,2}\omega_2.
\end{equation}

To obtain a tractable yet economically meaningful preference space, we specify consumer $i$'s type as $\theta_i=(\gamma,\cos t_i,\sin t_i)$, where $\gamma > 0$ and $t_i\in[-\pi,\pi]$ is independently drawn from the uniform distribution on this interval.

This parameterization encodes two substantive assumptions. First, the parameter $\gamma>0$ captures the relative importance of the vertical dimension compared with the horizontal dimensions in consumers' preferences. We model $\gamma$ as common across consumers and interpret it as a population-level taste for quality information: it summarizes how much the market as a whole rewards vertical content relative to horizontally differentiated content. This modeling choice gives us analytical tractability in studying how policy-relevant shifts in consumers' aggregate concern for quality affect equilibrium outcomes. Second, consumers differ in how they weight the two horizontal dimensions $(\omega_1,\omega_2)$, and this heterogeneity is summarized by a single parameter $t_i$. The pair $(\cos t_i,\sin t_i)$ places each consumer on the unit circle, in the spirit of circular spatial models such as \cite{salop1979}. Importantly, the circular geometry of the preference space is not imposed as an independent topological assumption. Rather, it emerges naturally from the two-component structure of the horizontal dimensions of $\omega$ together with the unit-norm restriction on horizontal preferences.

\subsection{Content Producers}
We next specify how content producers generate content. There are $N$ content producers, and each producer $n$ chooses an editorial strategy vector $b_n=(b_{n0},b_{n1},b_{n2})\in\mathbb{R}^3$,
subject to the capacity constraint $\|b_n\|\leq1$. The vector $b_n$ determines how the producer allocates its limited informational capacity across the three components of the underlying event state $\omega=(\omega_0,\omega_1,\omega_2)$. Shifting capacity toward a given component makes the producer's content more informative about that aspect of the event.

The capacity constraint reflects limitations on both the production and consumption sides of the content market. From the production side, the constraint $\|b_n\|\leq 1$ captures the idea that producing informative content requires scarce resources, including editorial attention, expertise, investigative effort, production time, screen space, and cognitive focus. From the consumption side, the constraint reflects the limited attention and cognitive processing capacity of consumers. Although AI can substantially reduce the cost of generating content, it does not eliminate the limited amount of information that can be effectively communicated and absorbed through a single content piece. Beyond a certain point, additional information overwhelms consumers and becomes less effective. Consequently, from the perspective of both the production and the consumption side, content producers must decide which aspects of an event to emphasize. Increasing the informativeness of content about one component necessarily reduces the attention and informational bandwidth available for the remaining components. 

If consumer $i$ consumes content from producer $n$, the information takeaway is given by the noisy signal
$$s_{in}(\omega,b_n)=b_n\cdot\omega+\varepsilon_{in},$$
where the errors $\varepsilon_{in}$ are independent across consumers and content producers, independent of $\omega$, and distributed as standard normals. The noise term captures imperfections in content production, interpretation, and consumption. Even when a producer emphasizes a particular aspect of the event, the consumer does not observe that component perfectly. 

We'd like to note that although each consumer observes a scalar signal, the signal should be interpreted as a bundled piece of content rather than as information about a single component of the event state. That is, a content producer’s content strategy $b_n$ determines the direction in the multidimensional state space along which its content is informative. Upon observing $s_{in}$, consumers update beliefs about all components of $\omega$, but only through the weighted index emphasized by the producer. This formulation captures the idea that content creators do not provide separate, independent measurements of every aspect of an event; rather, they synthesize selected aspects into a narrative. 

In equilibrium, the capacity constraint binds to maximize the value of information, so $\|b_n\|=1$. Without loss of generality, we normalize the sign of $b_n$ so that $b_{n0}\geq0$.\footnote{Because the noise $\varepsilon_{in}$ is symmetric about zero, $s_{in}(\omega,-b_n)$ has the same joint distribution with $\omega$ as $-s_{in}(\omega,b_n)$, an invertible relabeling of $s_{in}(\omega,b_n)$. Thus, $b_n$ and $-b_n$ are information-equivalent: they yield the same value of information for every consumer and hence the same prices and payoffs. If $b_{n0}<0$, the normalization therefore replaces $b_n$ with $-b_n$, reversing the sign of every component rather than only the first.} Given $\|b_n\|=1$ and this normalization, we can represent the producer's strategy in polar coordinates. Specifically, we write $b_n= \left(\sqrt{x_n},\sqrt{1-x_n}\cos t_n,\sqrt{1-x_n}\sin t_n\right)$, where $x_n\in[0,1]$ and $t_n\in[-\pi,\pi]$. This transformation gives the content strategy a clear economic interpretation. The scalar $x_n$ measures the share of the producer's information capacity allocated to the vertical component $\omega_0$. The parameter $t_n$ captures the producer's target direction in the horizontal information space. It identifies the mixture of $\omega_1$ and $\omega_2$ that the producer emphasizes. We refer to the pair $(x_n,t_n)$ as content producer $n$'s content strategy. 

\begin{figure}[ht]
\begin{center}
\caption{\centering Illustration of the Content Strategy}
\includegraphics[width=0.6\textwidth]{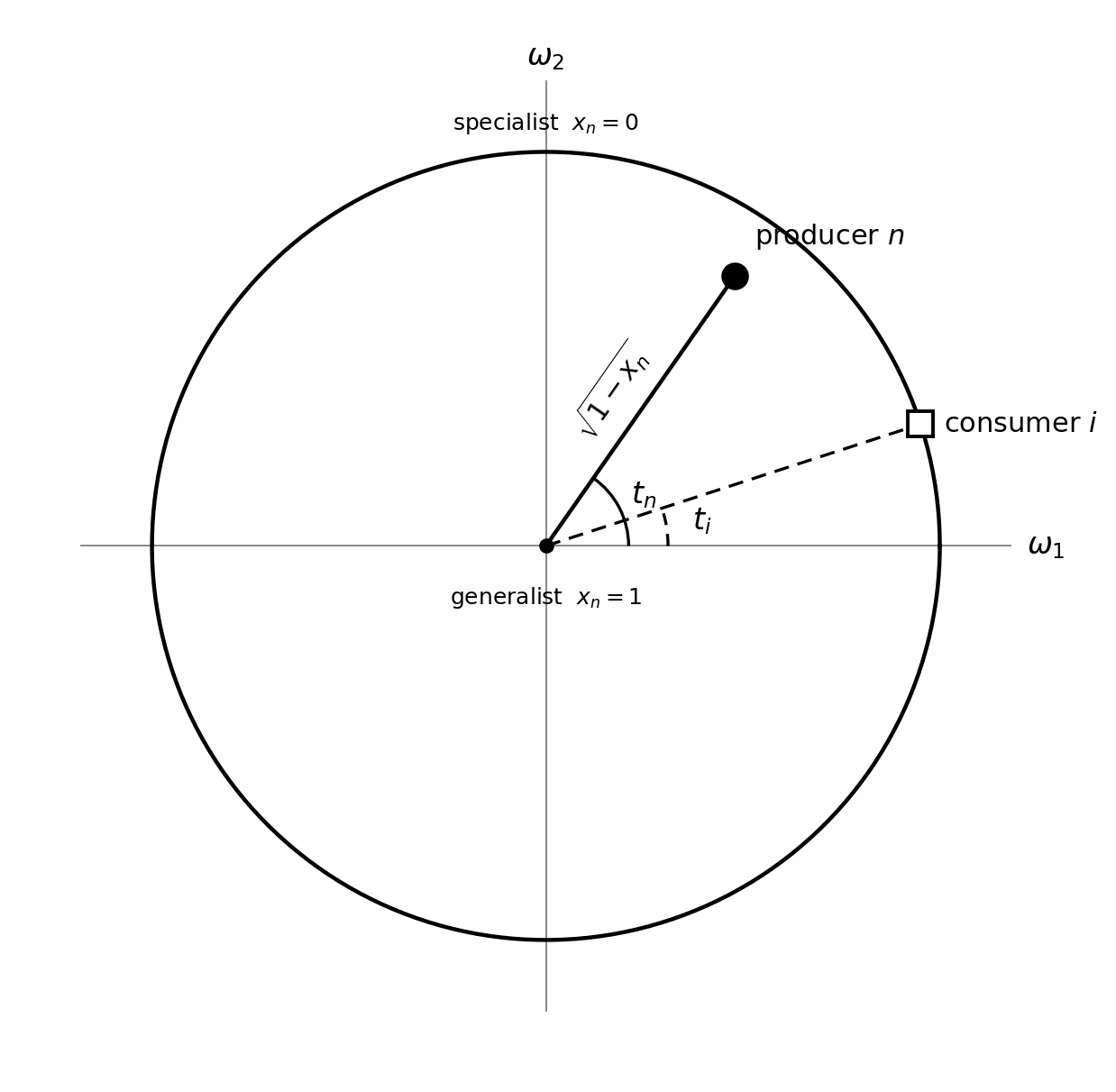}
\label{fig:unitcircle}
\end{center}
\end{figure}

Figure~\ref{fig:unitcircle} illustrates the geometry of our content strategy. Because the horizontal block of $b_n$ has norm $\sqrt{1-x_n}$, each producer can be represented as a single point \textit{inside} the unit disk. Its angular coordinate is the target direction $t_n$, and its distance from the center is $\sqrt{1-x_n}$. The square of this radius, $1-x_n$, is the share of content-production capacity the producer allocates to the horizontal dimensions. The center of the disk corresponds to a producer that devotes all of its capacity to the vertical component $\omega_0$ ($x_n = 1$); the rim corresponds to a producer that devotes all of its capacity to the horizontal dimension ($x_n = 0$). As $x_n$ falls, the producer moves outward from the center toward the rim, shifting capacity away from the vertical component and toward a sharper horizontal target. A consumer, by contrast, is represented as a point on the rim: her type $\theta_i = (\gamma, \cos t_i, \sin t_i)$ fixes her location $t_i$ on the preference circle. A producer with target $t_n$ creates content that is most aligned with consumers whose horizontal preference vector points in the same direction, so the alignment between producer $n$ and consumer $i$ is governed by the angular gap between $t_n$ and $t_i$. 

This circular geometry arises in the spirit of \cite{salop1979}, but the object a content producer chooses is richer than a firm's location in that model: whereas a Salop firm chooses only a point \textit{on} the circle, a producer here chooses a point \textit{in the interior} of the disk, characterized both by the target direction $t_n$ and by the radius $\sqrt{1-x_n}$. The target direction plays a role analogous to a Salop location. However, unlike a standard Salop framework in which firms differ only in their horizontal position, our model introduces an additional strategic margin through the radius of the producer’s location. Thus, our framework endogenizes the tradeoff between broadly valuable content and horizontally differentiated content, whereas previous spatial models typically impose this tradeoff exogenously.

\subsection{Role of AI}\label{subsec:aiaggregatormodel}
We next describe how we model AI information synthesis. In practice, AI search and related LLM-based systems retrieve information from a large collection of dispersed sources and synthesize it into a unified response to a user's query (e.g., \citealt{jdaisearch, alibabaaisearch, xiaohongshuaisearch, zhumeituan2026}). Our interest is not in the specific computational process through which this synthesis occurs, but in its economic consequence: the instrumental value of the information that AI ultimately provides to consumers. Because consumers use information to learn about the underlying event, we capture this instrumental value by how well the AI output helps consumers recover the underlying state $\omega$.

Specifically, we model AI as a reduced-form information technology using a truth-or-noise information structure \citep{lewissappington1994, johnsonmyatt2006}, characterized by a single reliability parameter $\alpha\in[0,1]$. With probability $\alpha$, AI successfully synthesizes the dispersed information environment and recovers the underlying state of the event $\omega$. With probability $1-\alpha$, synthesis fails and AI produces an uninformative output, modeled as an independent draw from the prior distribution of $\omega$. Consumers observe the AI output but not whether the synthesis process has succeeded.\footnote{This assumption can be interpreted as abstracting from costly verification by consumers. Consumers understand that AI is imperfect and know its reliability $\alpha$, but cannot costlessly determine whether a particular output is correct. Allowing consumers to verify AI-generated information at a cost would be a natural extension.} We assume that $\alpha$ is common knowledge among consumers and content producers. Thus, a higher $\alpha$ represents greater reliability of AI information synthesis and, consequently, greater instrumental information value provided by AI.

Two comments on our modeling approach are worth noting. First, our specification deliberately abstracts from the \emph{input side} of AI information synthesis and models the informativeness of its \emph{output} directly. In reality, an AI system retrieves information from multiple content sources and combines the signals contained in those sources to generate a response. One could therefore model this process explicitly by specifying which sources AI retrieves, how it weights their information, and how it reconciles conflicting signals.\footnote{Specifically, we can explicitly model AI synthesis as a mapping from the signals contained in individual content sources, for example, $s_{AI}=G(s_{i1},\ldots,s_{iN})$. Such an approach would require additional assumptions about source selection, weighting, and the treatment of conflicting signals. Our truth-or-noise specification instead describes AI's output directly, rather than deriving it from a particular synthesis rule applied to the producer signals.} We abstract from these details because our primary interest is not in the internal synthesis technology itself, but in how the information ultimately provided by AI affects consumers' information acquisition and content producers' strategic responses. Summarizing the outcome of this process with the single reliability parameter $\alpha$ allows us to transparently characterize how market equilibrium changes as AI's information capability improves.

Second, our specification allows AI to provide greater instrumental information value than any individual content source. As $\alpha$ approaches one, AI becomes increasingly effective at synthesizing dispersed information and recovering the underlying state. Emerging empirical evidence is consistent with this modeling approach. \citet{luo2025} find that LLMs outperform neuroscience experts in predicting experimental outcomes, with evidence that this advantage stems in part from their ability to integrate information across contexts. Similarly, \citet{asai2026} show that a retrieval-augmented LLM system that synthesizes evidence from a large scientific corpus can outperform expert-written responses.

We also consider two extensions to our modeling of AI information synthesis. First, the baseline model captures AI's potential to synthesize dispersed information into a highly informative output but assumes that the synthesis technology is common across consumers. Section~\ref{subsec:personalizedAI} relaxes this assumption by allowing AI to tailor its synthesis to individual consumers. Second, Section~\ref{subsec:endogenousAI} relaxes the exogeneity of AI reliability by allowing AI's information capability to depend explicitly on content producers' inputs, enabling us to examine the robustness of our main results when AI synthesis is linked to the underlying content it draws upon.

\subsection{Game Sequence and Equilibrium Concept}
The timing of the game is as follows. In the first stage, each content producer $n$ simultaneously chooses a content strategy $(x_n,t_n)$. In the second stage, consumer types $(t_i)$ and producers' content strategies become public information. After observing consumer types and the content strategies chosen by rival producers, each content producer $n$ sets a price $p_n(t_i)$. Content producers' prices are nonnegative, $p_n(t_i)\geq 0$, and access to AI-generated information is free. In the third stage, each consumer chooses at most one information source, either a single content producer or the AI, pays the corresponding price, and then observes the resulting information. If consumer $i$ chooses content producer $n$, she observes the signal $s_{in}(\omega,b_n)$. If instead she chooses the AI, she observes a realization about the event state governed by parameter $\alpha$. After observing the information and updating their beliefs, consumers make a downstream decision, such as whether to purchase a newly released product, adopt a new technology, or follow a health recommendation. We model this stage in reduced form by assuming that the consumer either takes the action or abstains. If the consumer takes the action, the payoff is $u(\omega,\theta_i)$, as defined in Equation~\eqref{eq:consumerutility}; otherwise, the payoff is zero. 

Several features of the game structure deserve comment. First, we model AI not as a strategic player, but as a new informational option available to consumers. By treating AI as a parameterized information technology, we isolate the economic question at the center of our analysis: how the availability of AI-mediated information synthesis reshapes consumers' choices and content producers' incentives, without introducing the separate question of how an AI provider would behave strategically. Second, content strategies are chosen before producers observe consumer types. The strategies $(x_n, t_n)$ should therefore be interpreted as long-run editorial commitments, that is, what an outlet investigates, and how it frames its coverage, and cannot be costlessly adjusted to the realized composition of demand. Third, our pricing stage connects to the literature on the pricing of information (\citealt{sellingcookies, designpriceinformation}). Permitting prices to condition on consumer types is a parsimonious way to capture the role of targeted advertising. Finally, we abstract from the details of the downstream decision and instead focus on the role of the information environment in shaping that decision. Whether the action is purchasing, adopting, or following, the common economic force is that consumers acquire information because it helps them make decisions whose payoffs depend on the underlying state of the world. Under this formulation, the value of information arises because better information improves the consumer's ability to identify situations in which taking the action is beneficial. 

The solution concept is a symmetric Perfect Bayesian Equilibrium (PBE). A symmetric PBE consists of (i) an editorial strategy $(x_n^*,t_n^*)$ and a pricing strategy $p_n^*(t_i)$ for each content producer $n=1,\ldots,N$; (ii) an information-source choice and a downstream action strategy for each consumer of type $t_i$; and (iii) a system of posterior beliefs that are updated according to Bayes' rule whenever possible. In equilibrium, each content producer chooses $(x_n,t_n)$ and $p_n(t_i)$ to maximize expected profits, taking the strategies of all other producers as given. Each consumer chooses the information source that maximizes her expected utility and, after observing the corresponding information, chooses the optimal downstream action. We focus on symmetric equilibria in which all content producers choose the same degree of focus on the vertical component, i.e., $x_n^*=x^*$ for all $n$, and the target directions $t_n^*$ are equally spaced over the horizontal space. Consumers follow identical decision rules conditional on their types. We solve the game by backward induction.

\section{Equilibrium Analysis} \label{sec:analysis}
We begin by characterizing consumers' decisions. We then derive content producers' equilibrium pricing and content production strategies. Finally, using the equilibrium characterization, we examine how the introduction of AI reshapes content production.

\subsection{Consumer Decisions}\label{subsec:consumeranalysis}
We begin with consumers' final-stage decisions. After choosing an information source, consumer $i$ observes the corresponding information output. If the consumer chooses content producer $n$, a noisy signal $s_{in}$ is observed. If the consumer chooses AI, an AI-generated synthesized content is observed, denoted by $a_i^{AI}$. Based on the observed information, consumer $i$ updates her belief about the payoff from taking the downstream action.

Let $\mathcal{I}_{iq}$ denote the information available to consumer $i$ after choosing source $q\in\{1,\ldots,N,AI\}$. Since not taking the action yields a normalized payoff of zero, consumer $i$ takes the downstream action if and only if
$$\mathbb{E}\left[u(\omega,\theta_i)\mid \mathcal{I}_{iq}\right]=\mathbb{E}\left[\theta_i\cdot\omega\mid \mathcal{I}_{iq}\right]>0.$$
When this posterior expectation equals zero, the consumer is indifferent, and we assume that she abstains. Before acquiring any information, the consumer's expected utility from taking the action is
$$\mathbb{E}[u(\omega,\theta_i)] = \mathbb{E}[\theta_i\cdot\omega] =0,$$
because each component of $\omega$ is independently distributed as a standard normal random variable. 

We next characterize consumers' information-source choices. In equilibrium, a consumer of type $t_i$ selects the information source that maximizes the value of information net of the price charged. Therefore, characterizing consumers' information acquisition decisions reduces to deriving the value of information associated with each source.

Let $v_{in}$ denote the value of information that content producer $n$ provides to consumer $i$, and let $v_{i,AI}$ denote the value of information provided by AI. Formally, the instrumental value of information is the ex ante increase in expected utility from making the downstream decision after observing the information provided by a source, relative to making the same decision without acquiring any information. The following lemma characterizes the value of information delivered by each source.

\begin{lemma}\label{lm:valueofinformation}
Define $\lambda = 1/(2\sqrt{\pi})$. The value of information that consumer $i$ obtains from content producer $n$ is
\begin{equation}\label{eqn:contentproducervoi}
    v_{in} = \lambda \left|\gamma\sqrt{x_n} + \sqrt{1-x_n}\cos(t_i-t_n) \right|.
\end{equation}
The value of information provided by AI is
\begin{equation}\label{eqn:aivoi}
    v_{i, AI} = v_{AI} = \alpha\lambda\sqrt{2(\gamma^2+1)}.
\end{equation}
\end{lemma}

Lemma \ref{lm:valueofinformation} provides a simple characterization of the value of information from each source. The value of information from a content producer consists of two components. The first component, $\gamma\sqrt{x_n}$, captures information about the vertical component of the event that is valuable to all consumers regardless of their horizontal taste. The second component, $\sqrt{1-x_n}\cos(t_i-t_n)$, captures information about the horizontal component. Within the readerships considered below, the term inside the absolute value in Equation~\eqref{eqn:contentproducervoi} is positive, so information value is highest at the producer's target direction and decreases with the horizontal distance $|t_i-t_n|$ from it.

In contrast, the value of information provided by AI is independent of consumer type. This reflects our deliberate modeling choice. In our baseline model, we focus on a powerful, broadly accessible AI technology without tailoring its output to individual consumers. By abstracting away from personalization, we isolate how improvements in AI's information synthesis capability affect content producers' strategic content choices. Consequently, as shown in Equation (\ref{eqn:aivoi}), every consumer derives the same expected informational benefit from AI. This value is increasing in both the relative importance of the vertical dimension compared to the horizontal dimension $\gamma$ and the probability of successful synthesis $\alpha$, implying that a more capable AI generates greater expected informational value for all consumers.

With Lemma \ref{lm:valueofinformation}, the consumer's information source choice can be characterized straightforwardly. Given a profile of content strategies and prices, $\left((x_n,t_n),p_n(t_i)\right)_{n=1}^N$, a consumer of type $t_i$ chooses the information source that provides the highest information value net of price. Formally, the consumer solves
\begin{equation}\label{eqn:maxmax}
\max\left\{\max_{n}\left\{v_{in}-p_n(t_i)\right\}, v_{AI}\right\}.
\end{equation}
If several sources attain the maximum in Equation~\eqref{eqn:maxmax}, the consumer selects one with the highest gross information value. Any remaining ties are broken by a fixed priority order over sources that ranks every content producer ahead of AI.

\subsection{Producer Strategies}\label{subsec:produceranalysis}
We next characterize the pricing stage of the game. After observing the profile of all content strategies $\left((x_n,t_n)\right)_{n=1}^N$, content producers simultaneously set type-specific prices. Since prices can be conditioned on consumer type, content producers compete \`a la Bertrand for each potential consumer independently.

Fix a profile of content strategies $\left((x_n,t_n)\right)_{n=1}^N$. For each consumer of type $t_i$, let
$$V_{i,-n} = \max_{m\neq n}v_{im}$$
denote the highest value of information provided by rival content producers. Because consumers may also choose AI, the best outside option available to consumer $i$ when considering content producer $n$ is
$$O_{in} = \max\{V_{i,-n},\,v_{AI}\}.$$

The following lemma characterizes the equilibrium pricing outcome.

\begin{lemma}\label{lem:pricing}
Given a profile of content strategies, the equilibrium price charged by content producer $n$ to consumer $i$ is
$$p_{n}^*(t_i) = \max\{v_{in}-O_{in},\,0\}.$$
Consequently, content producer $n$'s equilibrium profit is
$$\Pi_n = \mathbb{E}_{t_i}\!\left[p_{n}^*(t_i)\right].$$
\end{lemma}

\begin{proof}
Consider a consumer of type $t_i$. If $v_{in}<O_{in}$, then even a zero price cannot induce the consumer to purchase from content producer $n$. Hence, producer $n$ loses the consumer and earns zero profit. If instead $v_{in}\ge O_{in}$, then producer $n$ can guarantee the consumer by charging any price strictly below $v_{in}-O_{in}$. Under the tie-breaking rule following Equation~\eqref{eqn:maxmax}, Bertrand competition implies that the producer prices at the maximum level that still leaves the consumer indifferent between purchasing from producer $n$ and choosing the best alternative. Therefore, the equilibrium price is $p_{n}^*(t_i) = \max\{v_{in}-O_{in},0\}$.
Integrating over consumer types gives the producer's equilibrium profit.
\end{proof}

Lemma~\ref{lem:pricing} implies that a content producer captures a consumer if and only if it provides the highest net value of information among all available information sources. The equilibrium price equals the producer's informational advantage over the consumer's best outside option, allowing the pricing stage to be summarized entirely by differences in information values. Notice that the profit function contains no explicit production cost. Content producers allocate a fixed informational capacity, captured by the constraint $\|b_n\|\leq 1$, and we abstract from additional costs of reallocating this capacity or serving consumers.

Finally, we turn to the first stage, in which content producers simultaneously choose their content strategies before consumer types are realized. Since consumer types are independently and identically distributed, the expected profit of content producer $n$ under the profile $\left((x_n,t_n)\right)_{n=1}^N$ is
\begin{equation}
\Pi_n\!\left((x_n,t_n),(x_{-n},t_{-n})\right)=\mathbb{E}_{t_i}\!\left[p_{n}^*(t_i)\right].
\label{eqn:profit}
\end{equation}

The first stage therefore, defines a one-shot complete-information game among content producers. Although the realized consumer types are unknown when content strategies are chosen, all producers share the same prior distribution over consumer types and fully observe the game's primitives. Consequently, each producer chooses its content strategy to maximize expected profit in Equation~\eqref{eqn:profit}.

\subsection{Equilibrium Effects of AI}\label{subsec:aiaggregator}
In this section, we characterize the symmetric equilibrium of the model and compare it with the benchmark equilibrium without AI to illustrate how the introduction of AI affects content production. We begin by imposing two assumptions.

\begin{assumption}\label{ass:connected}
The number of content producers satisfies $N\geq \underline N(\gamma)$ where
$$\underline N(\gamma)\equiv
\begin{cases}
2, & \gamma\geq 1,\\
\min\left\{N\in\mathbb N:N\geq2,\,
\gamma^2(\pi/N)\geq \sin^3\left(\pi/N\right)
\right\}, & 0<\gamma<1.
\end{cases}$$
\end{assumption}

Assumption \ref{ass:connected} rules out disconnected readerships. The intuition is that when $0 < \gamma < 1$, a sufficiently horizontal signal is valuable to consumers both near a content producer's target direction and near the opposite direction. Throughout the analysis, we therefore focus on the case in which each content producer's readership forms a single connected interval.

\begin{assumption}\label{ass:alpha}
The AI's reliability parameter $\alpha$ satisfies
\begin{equation}
    \alpha < \frac{1}{\sqrt{2}}.
\end{equation}
\end{assumption}

Assumption \ref{ass:alpha} focuses the analysis on the parameter region in which AI influences consumers' information acquisition decisions without fully displacing content producers. This restriction ensures that a producer can offer a sufficiently well-matched consumer greater information value than AI, although, for consumers far from every producer's target direction, AI may be more valuable than any content actually offered. Theorem \ref{thm:AIequilibrium} characterizes the equilibrium of the model. The detailed proof is provided in the Appendix.

\begin{theorem}\label{thm:AIequilibrium}
Let $h=\pi/N$ and suppose that Assumptions~\ref{ass:connected} and~\ref{ass:alpha} hold. Define
$$
\underline{\alpha}(\gamma,N)
=\frac{\gamma^2 h+\sin h\,\cos 2h}{\sqrt{2(\gamma^2+1)\bigl(\gamma^2 h^2+\sin^2 h\bigr)}},
\quad
\overline{\alpha}(\gamma,N)
=\frac{\gamma^2 h+\sin h\,\cos h}{\sqrt{2(\gamma^2+1)\bigl(\gamma^2 h^2+\sin^2 h\bigr)}}.
$$
The equilibrium takes one of three forms, depending on the AI's reliability $\alpha$.
\begin{enumerate}
    \item If $0 < \alpha \leq \underline{\alpha}(\gamma,N)$, then:
    \begin{enumerate}
        \item Every content producer chooses vertical emphasis
        $$
            x_n^* = x^* = \frac{\gamma^2 h^2}{\gamma^2 h^2+\sin^2 h} \quad \text{for all } n,
        $$
        and the target directions $t_1^*,\dots,t_N^*$ are equally spaced with adjacent producers separated by an angular distance of $2h$.
        \item Content producer $n$ charges consumer $i$ a price $p_n^*(t_i)=\max\{v_{in} (x^*,t_n^*)-V_{i,-n}(x^*,t_{-n}^*),\,0\}$, where $V_{i,-n}(x^*,t_{-n}^*)=\max_{m\neq n} v_{im}(x^*,t_m^*)$. For consumers with $|t_i-t_n^*|\leq h$, this maximum is attained by an adjacent producer.
        \item Each consumer $i$ with $|t_i-t_n^*|\leq h$ chooses content producer $n$; no consumer chooses AI.
    \end{enumerate}
    \item If $\underline{\alpha}(\gamma,N) < \alpha \leq \overline{\alpha}(\gamma,N)$, then:
    \begin{enumerate}
        \item Content producers' content strategies and consumers' information acquisition choices are as in Case~1.
        \item There exists a unique cutoff $\hat{\delta}\in(0,h]$ such that producer $n$ charges consumer $i$ the price
        $$
            p_n^*(t_i)=
            \begin{cases}
                v_{in} (x^*,t_n^*) - v_{AI} & \text{if } |t_i-t_n^*|\leq \hat{\delta},\\
                v_{in} (x^*,t_n^*)-V_{i,-n}(x^*,t_{-n}^*) & \text{if } \hat{\delta}< |t_i-t_n^*|\leq h,\\
                0 & \text{otherwise},
            \end{cases}
        $$
        where $V_{i,-n}(x^*,t_{-n}^*)$ is defined as in Case~1.
    \end{enumerate}
    \item If $\overline{\alpha}(\gamma,N) < \alpha < 1/\sqrt{2}$, then:
    \begin{enumerate}
        \item Every content producer chooses vertical emphasis $x_n^* = x^*$ for all $n$; target horizontal directions are as in Case~1.
        \item Producer $n$ charges consumer $i$ the price $p_n^*(t_i) = \max\{ v_{in} (x^*,t_n^*)  - v_{AI},\,0\}$.
        \item There exists a unique threshold $\delta^*\in(0,h)$ such that consumer $i$ purchases from producer $n$ if and only if $|t_i-t_n^*|\leq \delta^*$; all remaining consumers choose AI.
        \item The pair $\bigl(x^*,\,\delta^*\bigr)$ is the unique solution to the system
        $$
            \cos\delta^* = \frac{\alpha\sqrt{2(\gamma^2+1)}-\gamma\sqrt{x^*}}{\sqrt{1-x^*}},
            \quad
            x^* = \frac{\gamma^2(\delta^*)^2}{\gamma^2(\delta^*)^2+\sin^2\delta^*}.
        $$
    \end{enumerate}
\end{enumerate}
\end{theorem}

Theorem~\ref{thm:AIequilibrium} characterizes how the equilibrium changes as AI becomes more reliable, and the theorem identifies three regimes. 

When AI reliability is low, $0<\alpha\leq \underline{\alpha}(\gamma,N)$, AI provides sufficiently low value of information that no consumer chooses it. Content producers therefore behave as they would in a world without AI. The key force determining both editorial choices and equilibrium prices is that, under Assumption~\ref{ass:connected}, each content producer competes only with its adjacent content producers and serves consumers whose locations satisfy $|t_i-t_n^*|\leq h$. One important feature of this AI-irrelevant equilibrium is the relationship between consumer location, prices, and surplus. Consumers located exactly at a producer's target direction receive the most information value from that producer. However, precisely because their alternatives provide less value of information, the producer can charge them the highest prices. As a result, these core consumers receive the lowest consumer surplus. By contrast, consumers located at the boundary between two adjacent content producers are indifferent between the two producers. Competition between adjacent producers then drives the price for these boundary consumers to zero. Thus, consumer surplus is lowest at the core of each producer's audience and highest at the fringe.

Once AI reliability crosses the lower threshold, $\underline{\alpha}(\gamma,N)<\alpha\leq \overline{\alpha}(\gamma,N)$, AI becomes attractive to some consumers. Given the consumer surplus pattern described above, the first consumers for whom the AI becomes relevant are not the fringe consumers, but rather the core consumers located close to a content producer's target direction. Content producers respond by lowering prices for these consumers near their target direction, just enough to retain them, as reflected in the pricing strategy in Case~2. The key feature of this region is that no consumer actually switches to AI. Content producers continue to serve the same audiences, and their editorial choices remain exactly as in Case~1. Thus, in this parameter range, AI operates purely as a competitive threat: consumers gain through lower prices, content producers lose margins, but the content itself remains unchanged.

Third, when AI reliability is high, $\overline{\alpha}(\gamma,N)<\alpha<1/\sqrt{2}$, AI becomes attractive enough that price reductions alone can no longer retain all consumers, and some consumers now choose AI over content producers. Consumers near the boundary between adjacent content producers are the first to leave, because they receive the lowest value of information from any content producer. Each content producer therefore serves only consumers sufficiently close to its own target direction, with the size of this retained audience summarized by the cutoff $\delta^*< h$. In this high-$\alpha$ regime, editorial choices respond because the retained audience has changed. The equilibrium pair $(x^*,\delta^*)$ is characterized by two conditions: the marginal consumer at distance $\delta^*$ is indifferent between purchasing from the content producer at a price of zero and using the AI, and each content producer chooses the editorial mix that is optimal for the consumers it still serves. Thus, once the AI changes the composition of a producer's audience, it also changes its editorial incentives.

The next proposition characterizes how content producers' editorial choices respond to the availability of AI. The proof is provided in the Appendix. 

\begin{proposition}\label{prop:AI_noAI_comparison}
Denote by $(x_0^*,t_{0n}^*)$ the symmetric equilibrium editorial choices in the benchmark economy without AI. Under Assumptions~\ref{ass:connected} and~\ref{ass:alpha}, the equilibrium editorial choices $(x^*,t_n^*)$ in the economy with AI compare with the benchmark as follows.
\begin{enumerate}
\item For every $\alpha\in(0,1/\sqrt{2})$, up to a common rotation and a relabeling of producers, the equilibrium target directions $t_n^*$ coincide with $t_{0n}^*$: they are equally spaced at angular distance $2h$ and depend only on $N$, not on $\alpha$ or $\gamma$.
\item If $\alpha\leq\overline{\alpha}(\gamma,N)$, then $x^*=x_0^*$. If $\alpha>\overline{\alpha}(\gamma,N)$, then $x^*<x_0^*$.
\item For $\alpha>\overline{\alpha}(\gamma,N)$, the gap $x_0^*-x^*$ is strictly increasing in $\alpha$ and strictly decreasing in $N$.
\item The threshold $\overline{\alpha}(\gamma,N)$ is strictly increasing in $\gamma$.
\end{enumerate}
\end{proposition}

Proposition~\ref{prop:AI_noAI_comparison} delivers the central message of our analysis: AI's information synthesis ability reshapes \emph{what} content producers make, and it does so along the vertical emphasis dimension alone. Content producers' horizontal positions are pinned down by the competitive pressure to spread out across the taste spectrum, and this force is unaffected by AI (Point~1). Vertical emphasis, by contrast, is a competitive instrument: in the benchmark economy, content producers invest in the vertical dimension partly to appeal to poorly matched consumers at the horizontal fringes of their audiences. Once the AI becomes reliable enough to capture precisely these fringe consumers, the instrument loses its purpose, and content producers redirect editorial resources toward the tastes of their remaining core audiences (Point~2). 

The remaining points characterize when this distortion bites hardest. It grows with the AI's reliability $\alpha$, as ever-better AI content peels away ever-better-matched consumers (Point~3). It shrinks with the number of content producers $N$, because crowded markets already induce horizontally oriented content, leaving less vertical emphasis investment for the AI to crowd out (Point~3). And it disappears entirely in markets where consumers' preference for the vertical dimension relative to the horizontal dimension, $\gamma$, is sufficiently strong, since the reliability threshold that the AI must clear to capture any consumers rises with consumers' preference for vertical quality (Point~4).

Although our model is deliberately stylized, early evidence from the deployment of AI search is consistent with its central mechanism. Three patterns stand out. First, AI search appears to remove disproportionately low-value, marginal readers, leaving a smaller but more engaged core audience. Using benchmark data from hundreds of publisher sites, Piano reports that search traffic fell by 36\% in 2025, while total audience fell by only 9\% and revenue by only 16\%.\footnote{\url{https://www.piano.io/resources/benchmark-update-what-happened-to-search-traffic-in-2025}} Moreover, the search visitors who disappeared were less valuable than those who remained: average new-subscription revenue per remaining search visitor rose by 31\%. This compositional shift is precisely the change in audience composition that drives our theoretical prediction. Second, publishers recognize this shift and are reorienting their objectives toward the remaining audience. \citet{reuters2026dnr} reports that, as AI search erodes referral traffic, publishers are increasingly seeking to ``maximise the engagement and value of the smaller but more loyal audience which continues to read, watch, and listen to news organizations' own outlets''. Third, this reorientation is reflected in content strategy. In a 2026 survey of 90 publishers conducted by Digiday and Braze, respondents reported emphasizing member-exclusive newsletters, proprietary research and data, podcasts, personalized content, and other formats designed to deepen relationships with existing subscribers.\footnote{\url{https://digiday.com/sponsored/the-state-of-subscriber-engagement-in-the-ai-era-publishers-first-party-relationship-strategies/}} Such content caters to the specific tastes of a committed core audience rather than to the broad audience, consistent with the shift in editorial emphasis our model predicts.

Our model highlights two policy implications that are easily overlooked when AI synthesizers are evaluated solely from the consumer side. First, assessments of AI should account for content producers’ endogenous supply-side responses. Any welfare evaluation of AI, and any regulatory intervention built on it, must account not only for the quality and accessibility of AI-generated answers today, but also for their equilibrium effect on content producers' incentive to produce the original content. Second, our results point to a potential lever for mitigating this distortion. Point~4 of Proposition~\ref{prop:AI_noAI_comparison} shows that when consumers place greater weight on the vertical dimension, that is, on the underlying facts and quality of an event relative to horizontal taste alignment, AI must become substantially more reliable before it begins to distort producers’ editorial choices. Although policymakers cannot directly choose consumers’ preferences, society’s average concern for the vertical dimension may be influenced through education, information campaigns, media literacy, quality certification, and marketing that makes factual accuracy and content quality more salient. Such efforts can strengthen consumers’ demand for vertically differentiated, high-quality content and thereby preserve producers’ incentives to invest in it. In this sense, one way to increase the resilience of original content production in the age of AI is not only to regulate the AI technology itself, but also to strengthen society’s demand for factual and quality-oriented content.

\section{Extensions}\label{sec:extensions}
In this section, we examine several aspects that are abstracted away from the main analysis. Our main analysis focuses on how the introduction of AI affects content production. However, greater polarization in the content produced does not necessarily imply greater divergence in consumers’ overall evaluations, as the content consumers ultimately encounter also depends on their choice of outlet. Accordingly, in Section \ref{subsec:consumerevaluation}, we examine how the AI synthesizer affects the alignment of consumer evaluations. In addition, as discussed in Section \ref{subsec:aiaggregatormodel}, our baseline model of AI abstracts from two features of real-world AI synthesizers. In Section \ref{subsec:personalizedAI}, we extend the model to allow AI to personalize its output to individual consumers. In Section \ref{subsec:endogenousAI}, we further endogenize the quality of the AI synthesizer’s output by allowing it to depend on the content supplied by content producers.

\subsection{Consumer Evaluation Alignment}\label{subsec:consumerevaluation}
We first define our measure of consumer evaluation alignment. Let $q^*(t_i)\in\{1,\ldots,N,AI\}$ denote consumer $i$'s equilibrium information source, and let
$$
z_i\equiv \mathbb E[\theta_i\cdot\omega\mid\mathcal I_{i,q^*(t_i)}]
$$
denote her posterior evaluation of the downstream action. We define consumer evaluation alignment as
$$
\mathcal S \equiv \mathbb E_{t_i,t_j}\left[\operatorname{Corr}(z_i,z_j\mid t_i,t_j)\right],
$$
where $i$ and $j$ are independently drawn from the unit mass of consumers. This measure captures the extent to which consumers' posterior evaluations move in the same direction across realizations of the underlying state. A higher value indicates greater co-movement in consumers' evaluations, even when their evaluation levels differ. Consumer evaluation alignment is a descriptive measure rather than a welfare criterion, since each consumer's payoff depends only on her own downstream decision.

Our main model specifies the AI's information structure from the perspective of an individual consumer. This is sufficient for characterizing consumers' information acquisition choices and the resulting equilibrium, but does not require specifying the joint distribution of AI information across consumers. Consumer evaluation alignment, in contrast, is a joint outcome and therefore requires such a specification. We impose the following assumption.

\begin{assumption}\label{ass:informationacrossconsumers}
Conditional on the underlying state $\omega$, the AI content realizations received by different consumers are independently generated. That is, for any $i\neq j$,
$$
a_i^{AI}\perp a_j^{AI}\mid\omega.
$$
Conditional on $\omega$, these realizations are also independent of every producer signal error $\varepsilon_{kn}$.
\end{assumption}

Assumption~\ref{ass:informationacrossconsumers} completes the joint information structure without changing the individual-level AI information structure in the main model. In particular, each consumer faces the same AI reliability $\alpha$ as before; the assumption only specifies how the realizations of this information technology are related across consumers and to producers' signal errors.

To compare consumer evaluation alignment with and without AI, let $\mathcal{S}_0$ denote consumer evaluation alignment in the benchmark without AI and $\mathcal{S}_{AI}$ denote consumer evaluation alignment in the presence of AI. Proposition \ref{prop:consumeralignment} summarizes the key result, with the proof provided in the Appendix.

\begin{proposition}\label{prop:consumeralignment}
Suppose Assumptions~\ref{ass:connected}, ~\ref{ass:alpha}, and~\ref{ass:informationacrossconsumers} hold. Then:
\begin{enumerate}
\item If $0<\alpha\leq\overline{\alpha}(\gamma,N)$, then
$$
\mathcal{S}_{AI}=\mathcal S_0.
$$
\item If $\overline{\alpha}(\gamma,N)<\alpha<1/\sqrt{2}$, then
$$
\mathcal S_{AI}<\mathcal S_0.
$$
\end{enumerate}
\end{proposition}

Proposition~\ref{prop:consumeralignment} translates the supply-side distortion of Proposition~\ref{prop:AI_noAI_comparison} into its consequence for consumers. When $\alpha\leq\overline{\alpha}(\gamma,N)$, all consumers continue to obtain information from content producers, and producers maintain the same content strategies as in the benchmark without AI. Consumer evaluation alignment therefore remains unchanged. Once $\alpha>\overline{\alpha}(\gamma,N)$, however, some consumers switch to AI, while content producers respond by reducing their emphasis on the vertical dimension that generates commonality in information across consumers. At the same time, under Assumption~\ref{ass:informationacrossconsumers}, consumers who use AI receive independently generated information realizations conditional on the underlying state, limiting the commonality of their posterior evaluations. Together, these changes reduce the extent to which consumers' evaluations co-move, resulting in lower consumer evaluation alignment than in the benchmark without AI. Thus, the greater horizontal differentiation in content production identified in our main analysis is accompanied by greater divergence in consumers' evaluations once AI is sufficiently reliable to attract consumers.

\subsection{Personalized AI Synthesizer}\label{subsec:personalizedAI}
In this section, we consider an AI synthesizer that can tailor its output to each individual consumer's type. We retain the synthesis stage of the main model and add a single delivery stage on top of it. As before, the AI's synthesized content $a_i^{AI}$ recovers $\omega$ with probability $\alpha$, and with probability $1-\alpha$ is an independent draw from the prior distribution of $\omega$. Unlike in the main model, however, consumer $i$ does not observe $a_i^{AI}$ directly. Instead, given her preference vector $\theta_i=(\gamma,\cos t_i,\sin t_i)$, AI delivers the personalized signal
$$
s_i^{AI} =  \frac{\theta_i\cdot a_i^{AI}}{\sqrt{\gamma^2+1}}  +  \varepsilon_i^{AI},
$$
where $\varepsilon_i^{AI}\sim N(0,1)$ is an idiosyncratic noise term, independent across consumers and all other random variables in the model, and the factor $\sqrt{\gamma^2+1}=\lVert\theta_i\rVert$ normalizes the signal direction to unit length, so that the synthesizer's signal is directly comparable to content producers' signals under the capacity constraint of the main model.

Our signal structure embodies personalization in its purest form. Because AI projects its synthesis onto consumer $i$'s own preference vector $\theta_i$, it weights each dimension of the event in exact proportion to the importance that consumer $i$ attaches to it. In the language of the main model, conditional on successful synthesis, the signal delivered to consumer $i$ corresponds to the content strategy
$$
(x,t)  =  \left(\frac{\gamma^2}{\gamma^2+1},  t_i\right),
$$
that is, to consumer $i$'s \emph{ideal} content: the vertical emphasis $\gamma^2/(\gamma^2+1)$ replicates the weight her preferences place on the vertical dimension, and the horizontal direction is aimed exactly at her taste location $t_i$, so that no capacity is spent on horizontal directions she does not value. This is precisely the strategy that maximizes her information value under the capacity constraint.

The following lemma characterizes the value of information provided by the personalized AI synthesizer, and the proof is provided in the Appendix.

\begin{lemma}\label{lm:personalizedAI}
The value of information provided by the personalized AI synthesizer is
$$
v_{i,AI}^{\mathrm{pers}} = v_{AI}^{\mathrm{pers}} = \alpha\lambda\sqrt{\gamma^2+1}.
$$
\end{lemma}

Lemma~\ref{lm:personalizedAI} shows that personalization preserves a key feature of AI in our main model: although the signal itself is tailored to each consumer's type, its ex ante information value is the same across consumers. Intuitively, each consumer receives a signal directed exactly toward her own preference vector. Since all consumers have the same preference-vector norm, $\|\theta_i\|=\sqrt{\gamma^2+1}$, their ideal signals provide the same ex ante information value. The personalized AI therefore continues to provide a type-independent outside option against which content producers compete.

The personalized AI synthesizer can also be mapped directly into the information structure of the main model. Recall that the value of AI information in the main model is $v_{AI} = \alpha\lambda\sqrt{2(\gamma^2+1)}$.
Under personalization, it becomes $v_{AI}^{\mathrm{pers}} = \alpha\lambda\sqrt{\gamma^2+1} = \frac{\alpha}{\sqrt2}\lambda\sqrt{2(\gamma^2+1)}$.
Thus, from the perspective of consumers' information acquisition choices and content producers' competitive incentives, the personalized AI is equivalent to the AI in the main model with reliability $\alpha/\sqrt2$. Consequently, the equilibrium characterization and the mechanisms established in the main analysis continue to apply after replacing $\alpha$ with $\alpha/\sqrt2$. The adjustment in effective reliability reflects the additional unit-variance delivery noise.

\subsection{Endogenous AI Reliability}\label{subsec:endogenousAI}
Our main model treats the reliability of AI as exogenous. This abstraction allows us to isolate how the introduction of an AI synthesizer affects content producers' incentives, but it abstracts from a potential feedback in the opposite direction: the quality of AI-generated information may itself depend on the content available from original content producers. Consequently, changes in the information supplied by content producers may also affect the quality of the AI's output.

To capture this feedback, we allow AI reliability to depend on the amount of vertical information available in the content-provider ecosystem. Let
$$\bar{x} \equiv \frac{1}{N}\sum_{n=1}^{N}x_n$$
denote the average vertical emphasis of content producers. The AI's realized reliability is
$$\alpha(\bar{x};\eta),$$
where $\eta$ indexes the underlying capability of the AI technology.

\begin{assumption}\label{ass:sourcereliability}
For each $\eta\in[0,1/\sqrt{2})$, $\alpha(\bar{x};\eta)$ is nonnegative, weakly increasing, and twice continuously differentiable in $\bar{x}\in[0,1]$, with $\alpha(1;\eta)=\eta$. Moreover, $\alpha(\bar{x};\eta)$ is continuously differentiable in $\eta$, with $\alpha_\eta(\bar{x};\eta)>0$.
\end{assumption}

The parameter $\eta$ captures the underlying capability of the AI technology: when content producers devote all of their capacity to the vertical dimension, AI achieves reliability $\eta$. The condition $\alpha_x(\bar{x};\eta)\geq 0$ captures the source dependence of AI reliability, whereby greater production of vertical information by content producers improves the reliability of AI. We refer to the technology as strictly source-dependent at $\bar{x}$ whenever $\alpha_x(\bar{x};\eta)>0$. The main model is nested as the special case $\alpha(\bar{x};\eta)\equiv\eta$, under which AI reliability is independent of content producers' content choices.

We use average vertical emphasis, rather than the total amount of vertical content, so that AI reliability depends on the composition of content supplied by content producers rather than mechanically on the number of producers. This specification deliberately abstracts from other features of the original content environment, such as horizontal diversity and correlation across sources, allowing us to isolate the feedback between producers' provision of vertical information and the reliability of the AI synthesizer. All other elements of the model remain unchanged.

Following the main model, we can derive the value of information provided by AI remains $v_{AI}(\bar{x};\eta) = \alpha(\bar{x};\eta)\lambda\sqrt{2(\gamma^2+1)}$. The key difference from the main model is that content producer $n$, when choosing $x_n$, now recognizes that
$$\frac{\partial\alpha(\bar{x};\eta)}{\partial x_n} = \frac{1}{N}\alpha_x(\bar{x};\eta).$$
Thus, a producer's content choice affects not only the value of its own information, but also the quality of the AI synthesizer against which it competes.

As in the main analysis, we focus on symmetric equilibria in which content producers' target directions are equally spaced. In this extension, we also maintain the adjacent-rival structure of the no-AI benchmark.\footnote{Specifically, for each consumer in a producer's equilibrium readership, the best rival value is attained by an adjacent producer for whom the term inside the absolute value in Equation~\eqref{eqn:contentproducervoi} is nonnegative. Assumption~\ref{ass:connected} guarantees this structure at $x_0^*$. Because endogenous reliability may lower equilibrium vertical emphasis below $x_0^*$, we maintain this structure directly here.} Theorem~\ref{thm:endogenousAI} characterizes the three equilibrium cases, paralleling those in Theorem~\ref{thm:AIequilibrium}, together with the corresponding equilibrium conditions and comparisons of vertical emphasis. The proof is provided in the Appendix.

\begin{theorem}\label{thm:endogenousAI}
Suppose Assumptions~\ref{ass:connected}, \ref{ass:alpha}, and~\ref{ass:sourcereliability} hold. Any interior symmetric equilibrium takes one of the following three forms.

\begin{enumerate}
    \item If AI does not bind any consumer's outside option, then
    $$
    x^*=x_0^*.
    $$
    All consumers obtain information from content producers, and equilibrium content strategies, prices, and information acquisition choices coincide with those in the benchmark without AI.

    \item Suppose all consumers continue to obtain information from content producers, but AI is the binding outside option for consumers within distance $\hat{\delta}\in(0,h]$ of each producer's target direction. Then the equilibrium pair $(x^*,\hat{\delta})$ satisfies
    $$
    \sqrt{2(\gamma^2+1)}\,\alpha(x^*;\eta) = \gamma\sqrt{x^*} + \sqrt{1-x^*}\cos(2h-\hat{\delta}),
    $$
    and
    $$
    \frac{\gamma h}{2\sqrt{x^*}} - \frac{\sin h}{2\sqrt{1-x^*}} = \frac{\sqrt{2(\gamma^2+1)}\,\hat{\delta}}{N}\alpha_x(x^*;\eta).
    $$
    Moreover,
    $$
    x^*\leq x_0^*,
    $$
    with strict inequality whenever $\alpha_x(x^*;\eta)>0$.

    \item Suppose AI attracts a positive share of consumers. Then there exists a unique cutoff $\delta^*\in(0,h)$ such that consumers within distance $\delta^*$ of a producer's target direction obtain information from that producer, while all remaining consumers use AI. The equilibrium pair $(x^*,\delta^*)$ satisfies
    $$
    \gamma\sqrt{x^*} + \sqrt{1-x^*}\cos\delta^* = \sqrt{2(\gamma^2+1)}\,\alpha(x^*;\eta),
    $$
    and
    $$
    \frac{\gamma\delta^*}{2\sqrt{x^*}} - \frac{\sin\delta^*}{2\sqrt{1-x^*}} = \frac{\sqrt{2(\gamma^2+1)}\,\delta^*}{N}\alpha_x(x^*;\eta).
    $$
    Moreover,
    $$
    x^*<x_0^*.
    $$
\end{enumerate}
\end{theorem}

Theorem~\ref{thm:endogenousAI} shows that our main content-polarization result is robust to allowing AI reliability to depend endogenously on content production. Whenever the AI affects the market, content producers weakly reduce their vertical emphasis relative to the benchmark without AI, with a strict reduction once the AI attracts consumers. Moreover, source-dependent AI reliability introduces an additional strategic force that can further intensify content polarization. In the main model, producers reduce their vertical emphasis after AI adoption because the loss of marginal consumers shrinks their audiences, increasing the relative return to horizontally targeted content. With endogenous AI reliability, producers have an additional incentive to reduce vertical emphasis: doing so reduces the amount of commonly valuable information available to the AI, lowers its reliability, and thereby weakens it as a competitive outside option. Importantly, this strategic reliability effect can arise even before consumers actually adopt AI, as long as the AI constrains content producers' pricing, as in Case 2. Thus, allowing AI reliability to depend on original content production not only preserves the polarization mechanism identified in the main model, but also introduces an additional force toward polarization.

\section{Conclusion} \label{sec:conclusion}
Our paper studies how AI information synthesis reshapes content markets when content producers respond strategically to its availability. We show that the direct informational benefit of AI need not translate straightforwardly into a less polarized information environment. As AI becomes sufficiently reliable to attract consumers, it changes the audience that content producers compete to serve and, consequently, how they allocate their limited informational capacity. Producers shift away from information that is broadly valuable across consumers and toward more horizontally differentiated content, increasing polarization in the remaining producer-generated content. Understanding the consequences of AI information synthesis therefore requires considering not only what AI produces, but also how the surrounding information ecosystem adapts to it.

Our analysis abstracts from the entry and exit of content producers and takes the number of producers as given. This assumption allows us to isolate how AI competition affects the editorial choices of existing producers, but AI may also reshape the content market along the extensive margin. Incorporating endogenous market structure would therefore be a natural direction for future research.

\singlespacing

\bibliographystyle{chicago} 
\bibliography{literature}

\newpage 
\doublespacing
\appendix
\setcounter{lemma}{0}
\renewcommand{\thelemma}{A\arabic{lemma}}
\begin{center}
    {\Large Appendix}
\end{center}
\setcounter{equation}{0} 
\renewcommand{\theequation}{A\arabic{equation}}

\subsection*{Proof of Lemma \ref{lm:valueofinformation}:}
\vspace{-0.1in}
Recall that the value of information is the ex ante increase in expected utility from making the downstream decision after observing information from a source, relative to making the same decision without acquiring any information.

We first derive the value of information provided by a content producer. Suppose consumer $i$ chooses content producer $n$ and observes signal $s_{in}$. Since both the payoff $\theta_i\cdot\omega$ and the signal $s_{in}$ are jointly normally distributed, the consumer updates the beliefs by normal projection. Because each component of $\omega$ is independently distributed as a standard normal, $\mathrm{Cov}(\theta_i\cdot\omega,s_{in}) =\theta_i\cdot b_n$ and $\mathrm{Var}(s_{in}) =\|b_n\|^2+\mathrm{Var}(\varepsilon_{in}) =2$, where the last equality follows from $\|b_n\|=1$. Therefore, the posterior expectation of the payoff is
\begin{equation}
\mathbb{E}\!\left[\theta_i\cdot\omega\mid s_{in}\right] = \frac{\theta_i\cdot b_n}{2}s_{in}.
\label{eqn:posterior}
\end{equation}

Consumer $i$ takes the downstream action if and only if $\mathbb{E}\!\left[\theta_i\cdot\omega\mid s_{in}\right]>0.$ Hence, the ex ante value of information from content producer $n$ is
$$v_{in} = \mathbb{E}\!\left[\max\!\left(\{\mathbb E[\theta_i\cdot\omega\mid s_{in}],0\right\}\right].$$

By Equation~\eqref{eqn:posterior}, the posterior expectation is itself normally distributed with mean zero and standard deviation $\sigma = |\theta_i\cdot b_n|/\sqrt{2}$. Using the identity $\mathbb{E}[\max(Z,0)] = \sigma/\sqrt{2\pi}$ for any mean-zero normal random variable $Z\sim N(0,\sigma^2)$, we obtain
$$
v_{in}=\frac{|\theta_i\cdot b_n|}{2\sqrt{\pi}}.
$$

Substituting $\theta_i = (\gamma,\cos t_i,\sin t_i)$ and $b_n =\left(\sqrt{x_n},\sqrt{1-x_n}\cos t_n,\sqrt{1-x_n}\sin t_n\right)$, gives
$$
\theta_i\cdot b_n = \gamma\sqrt{x_n} + \sqrt{1-x_n} \left(\cos t_i\cos t_n + \sin t_i\sin t_n\right) = \gamma\sqrt{x_n} +\sqrt{1-x_n}\cos(t_i-t_n).
$$

Defining $\lambda=1/(2\sqrt{\pi})$, the value of information from content producer $n$ becomes
$$
v_{in} = \lambda \left|\gamma\sqrt{x_n} + \sqrt{1-x_n}\cos(t_i-t_n) \right|.
$$

If the consumer chooses AI, the consumer observes an AI-generated output about the event state, but does not observe whether AI has successfully recovered the true state or has instead generated pure noise. With probability $\alpha$, $a_i^{AI} = \omega$; with probability $1-\alpha$, $a_i^{AI} = \zeta$, where $\zeta$ is an independent draw from the prior distribution of $\omega$. The realization of $a_i^{AI}$ carries no information about the success event, so the posterior probability of success given the output equals $\alpha$. Averaging over the two possibilities,
$$\mathbb{E}\left[\theta_i \cdot \omega \mid a_i^{AI}\right]= \alpha\,\mathbb{E}\left[\theta_i \cdot \omega \mid a_i^{AI},\text{success}\right] + (1-\alpha)\,\mathbb{E}\left[\theta_i \cdot \omega \mid a_i^{AI},\text{failure}\right]= \alpha\,\theta_i \cdot a_i^{AI}.$$
The consumer therefore takes the downstream action if and only if $\theta_i \cdot a_i^{AI} > 0$. Ex ante, $\theta_i \cdot a_i^{AI}$ equals $\theta_i \cdot \omega$ under success and $\theta_i \cdot\zeta$ under failure, both distributed $N(0, \|\theta_i\|^2)$, so the posterior expectation is a mean-zero normal random variable with variance $\alpha^2\|\theta_i\|^2 = \alpha^2(\gamma^2+1)$. Using the half-normal identity, the value of information from AI is
$$
v_{i,{AI}} = \frac{\alpha\sqrt{\gamma^2+1}}{\sqrt{2\pi}} = \alpha\lambda\sqrt{2(\gamma^2+1)}.
$$

\subsection*{Proof of Theorem \ref{thm:AIequilibrium}:}
By Equation \eqref{eqn:aivoi}, the value of information provided by AI is increasing in the reliability parameter $\alpha$. Hence, when $\alpha$ is sufficiently small, corresponding to Case~1 of Theorem~\ref{thm:AIequilibrium}, AI does not bind as a competitive alternative option for consumers. In this case, content producers compete as if AI were absent. 

Consider a representative content producer $n$ with target direction $t_n^*=0$. In a symmetric equilibrium, the target directions $t_1^*,\dots,t_N^*$ are equally spaced with adjacent producers separated by an angular distance of $2h$. Also, all other content producers choose the same vertical component $x$. One can verify that, under Assumption~\ref{ass:connected}, each content producer competes only with its adjacent content producers. Therefore, by Lemma~\ref{lem:pricing}, the pricing strategies and consumers' information-acquisition choices are as described in Case~1 of Theorem~\ref{thm:AIequilibrium}.

Consider a unilateral deviation in $x_n$, holding the other producers' choices fixed at $x$. Let $O_n(t; x)$ denote the value of type $t$'s best alternative to producer $n$ (independent of $x_n$, since it depends only on AI and the other producers). Under Assumption~\ref{ass:connected}, the readership is a connected interval. Moreover, at the symmetric rival profile, both $v_n(t; x_n)$ and $O_n(t; x)$ are symmetric around $t = 0$. The readership can therefore be written as $[-\rho_n(x_n; x),\, \rho_n(x_n; x)]$, whose right-hand endpoint $\rho_n(x_n; x)$ moves with the deviation. By continuity, the marginal consumer satisfies
\begin{equation}
v_n\left(\rho_n(x_n; x); x_n\right) = O_n\left(\rho_n(x_n; x); x\right).
\label{eqn:marginal}
\end{equation}
By symmetry, producer $n$'s expected profit is
\begin{align*}
\Pi_n(x_n; x) &= \frac{1}{\pi}\int_0^{\rho_n(x_n; x)}\left[v_n(t; x_n) - O_n(t; x)\right] dt \\
&= \frac{1}{\pi}\int_0^{\rho_n(x_n; x)} \left[\lambda\left(\gamma\sqrt{x_n} + \sqrt{1-x_n}\cos t\right) - O_n(t; x)\right] dt.
\end{align*}
The absolute-value operator in Equation~\eqref{eqn:contentproducervoi} can be omitted in the second equality under Assumption~\ref{ass:connected}.

Differentiating $\Pi_n(x_n; x)$ with respect to $x_n$ by Leibniz's rule gives
$$\frac{\partial \Pi_n(x_n; x)}{\partial x_n} = \frac{1}{\pi}\int_0^{\rho_n(x_n; x)}
\frac{\partial v_n(t; x_n)}{\partial x_n}\, dt+ \frac{1}{\pi}\left[v_n(\rho_n; x_n) - O_n(\rho_n; x)\right]\frac{\partial \rho_n(x_n; x)}{\partial x_n}.$$
The second term is zero by Equation~\eqref{eqn:marginal}. Therefore,
\begin{equation}
\frac{\partial \Pi_n(x_n; x)}{\partial x_n}= \frac{\lambda}{\pi}\int_0^{\rho_n(x_n; x)}
\left[\frac{\gamma}{2\sqrt{x_n}} - \frac{\cos t}{2\sqrt{1-x_n}}\right] dt.
\label{eqn:FOCdev}
\end{equation}
In a symmetric equilibrium of Case~1, $x_n = x$; moreover, the AI synthesizer does not bind, so $O_n(t; x)$ is the value offered by the adjacent producer and, by the symmetry between adjacent producers, $\rho_n(x; x) = h$. The first-order condition is therefore
$$\left.\frac{\partial \Pi_n(x_n; x)}{\partial x_n}\right|_{x_n = x}= \frac{\lambda}{\pi}\int_0^h\left[\frac{\gamma}{2\sqrt{x}} - \frac{\cos t}{2\sqrt{1-x}}\right] dt= \frac{\lambda}{\pi}\left[\frac{\gamma h}{2\sqrt{x}}- \frac{\sin h}{2\sqrt{1-x}}\right] = 0.$$
Hence,
$$x^* = \frac{\gamma^2 h^2}{\gamma^2 h^2 + \sin^2 h}.$$

Because the readership cutoff varies with $x_n$, the second derivative must account for the induced change in $\rho_n(x_n; x)$. Differentiating Equation~\eqref{eqn:FOCdev} once more, with $\partial\rho_n/\partial x_n$ obtained from Equation~\eqref{eqn:marginal} by the implicit function theorem (the adjacent producer being the binding outside option at the marginal consumer), and substituting the first-order condition yields, at the symmetric candidate profile,
$$\left.\frac{\partial^2 \Pi_n(x_n; x)}{\partial x_n^2}\right|_{x_n = x = x^*}= \frac{\lambda}{8\pi h^2 (1-x^*)^{3/2}\sin h} \left[(\sin h - h\cos h)^2 - \frac{2h^2\sin^2 h}{x^*}\right] < 0.$$
The inequality follows from $0 < \sin h - h\cos h < h\sin h$, which holds because $h = \pi/N \le \pi/2$, together with $0 < x^* < 1$. Thus, the first-order condition characterizes producer $n$'s optimal vertical dimension choice in Case~1.

Notice that, in the equilibrium in which content producers compete as if the AI were absent, consumers located exactly at a producer's target direction are charged the highest price. In particular, a consumer with $t_i=t_n^*$ has the strongest preference for content producer $n$, allowing producer $n$ to extract the largest incremental value relative to the consumer's best adjacent alternative. By contrast, consumers located at the boundary between two adjacent content producers, $t_i=t_n^*\pm h$, are indifferent between the two producers and are therefore charged the lowest price, namely zero.

This pricing pattern also implies that consumer surplus is lowest for consumers located at $t_i=t_n^*$. For such a consumer, the surplus from purchasing from content producer $n$ is equal to the value of the best adjacent alternative:
$$v_{in}(x^*,t_n^*)-p_n^*(t_n^*) = V_{i,-n}(x^*,t_{-n}^*) =\lambda\left( \gamma\sqrt{x^*} + \sqrt{1-x^*}\cos(2h)\right).$$
Here, the adjacent content producers are located at distance $2h$ from $t_n^*$. By contrast, a consumer located at the boundary $t_i=t_n^*\pm h$ receives surplus
$$\lambda\left(\gamma\sqrt{x^*} + \sqrt{1-x^*}\cos h\right).$$
Since $\cos h>\cos(2h)$, the boundary consumer obtains higher surplus than the consumer located exactly at the producer's target direction. Therefore, the consumers located at $t_i=t_n^*$ are the first to be attracted by AI as $\alpha$ increases.

It follows that the threshold $\underline{\alpha}(\gamma,N)$ is determined by equating the value of AI to the consumer surplus of a consumer located at $t_i=t_n^*$:
$$
\underline{\alpha}(\gamma,N)\sqrt{2(\gamma^2+1)}
= \gamma\sqrt{x^*} + \sqrt{1-x^*}\cos(2h).
$$
Equivalently,
$$
\underline{\alpha}(\gamma,N) = \frac{\gamma\sqrt{x^*}+\sqrt{1-x^*}\cos(2h)}{\sqrt{2(\gamma^2+1)}} = \frac{\gamma^2 h+\sin h\,\cos 2h}{\sqrt{2(\gamma^2+1)\bigl(\gamma^2 h^2+\sin^2 h\bigr)}}.
$$

We next consider Case~2 of Theorem~\ref{thm:AIequilibrium}. Once $\alpha>\underline{\alpha}(\gamma,N)$, the AI synthesizer becomes an attractive alternative for consumers who obtain relatively low surplus from content producers. As discussed above, these consumers are located closest to a content producer's target direction, namely at $t_i=t_n^*$. By the pricing rule in Lemma~\ref{lem:pricing}, content producers respond by lowering prices in order to retain these consumers.

More specifically, for a given $\alpha\in(\underline{\alpha}(\gamma,N),\overline{\alpha}(\gamma,N)]$, there exists a unique cutoff $\hat{\delta}\in(0,h]$ such that the AI synthesizer binds as the relevant alternative option for consumers with $|t_i-t_n^*|\leq \hat{\delta}$, while the best adjacent content producer remains the relevant alternative option for consumers with $\hat{\delta}<|t_i-t_n^*|\leq h$. Thus, for consumers close to producer $n$'s target direction, content producer $n$ lowers its price relative to Case~1. For consumers closer to the boundary between adjacent producers, the pricing rule is unchanged from Case~1.

The cutoff $\hat{\delta}$ is determined by the consumer who is exactly indifferent between using AI and relying on the best adjacent content producer as the alternative option. On the right half of producer $n$'s market, this condition is
$$
\alpha\sqrt{2(\gamma^2+1)} = \gamma\sqrt{x^*} + \sqrt{1-x^*}\cos(2h-\hat{\delta}).
$$
Since the right-hand side is increasing in $\hat{\delta}$ on $[0,h]$, the cutoff is unique. 

We next show that the price reductions near each target direction leave producers' vertical emphasis unchanged. Consider first the interior of Case~2, where $0<\hat{\delta}<h$. As in Case~1, fix producer $n$'s target direction at zero and consider a unilateral change in $x_n$, holding AI reliability fixed and letting all rivals choose a common vertical emphasis $x$. Let $V_{-n}(t;x)$ denote the value offered to type $t$ by the best rival producer, which on the right half of producer $n$'s market equals $\lambda(\gamma\sqrt{x}+\sqrt{1-x}\cos(2h-t))$, and let $O_n(t;x)=\max\{V_{-n}(t;x),v_{AI}\}$. Written at this rival profile, the cutoff condition above becomes $V_{-n}(\hat{\delta};x)=v_{AI}$ and defines $\hat{\delta}=\hat{\delta}(x;\alpha)$. Since producer $n$'s own value of information does not enter this condition, $\hat{\delta}$ does not depend on $x_n$.

At the symmetric profile, producer $n$ strictly outperforms AI throughout $[-\hat{\delta},\hat{\delta}]$, and, since $\hat{\delta}<h$, the adjacent producers strictly outperform AI at the readership boundaries $\pm h$. By continuity, both strict inequalities persist after sufficiently small changes in $x_n$; the second continues to hold at the perturbed readership boundaries, which remain close to $\pm h$. Producer $n$ therefore retains every consumer whose price is constrained by AI, and its marginal consumers continue to be priced against the adjacent producers. Hence, under any such deviation, producer $n$'s readership is the same as in the economy without AI. Denote this common readership by $[-\rho_n(x_n;x),\rho_n(x_n;x)]$.

Let $\Pi_n^0(x_n;x)$ denote producer $n$'s profit without AI, holding rivals' strategies fixed at the same profile. For the deviations just described, Lemma~\ref{lem:pricing} gives
\begin{align*}
    \Pi_n(x_n;x)&=\frac{1}{\pi}\int_0^{\hat{\delta}}\left[v_n(t;x_n)-v_{AI}\right]\,dt+\frac{1}{\pi}\int_{\hat{\delta}}^{\rho_n(x_n;x)}\left[v_n(t;x_n)-V_{-n}(t;x)\right]\,dt\\
&=\Pi_n^0(x_n;x)-C_n(x;\alpha),
\end{align*}
where
$$C_n(x;\alpha)\equiv\frac{1}{\pi}\int_0^{\hat{\delta}(x;\alpha)}\left[v_{AI}-V_{-n}(t;x)\right]\,dt.
$$
The term $C_n(x;\alpha)$ is the aggregate price concession induced by AI. Neither its integrand nor its integration limit depends on $x_n$. Thus, AI locally subtracts a constant from producer $n$'s profit without changing the relative profitability of its editorial choices.

In particular,
$$\frac{\partial\Pi_n(x_n;x)}{\partial x_n}=\frac{\partial\Pi_n^0(x_n;x)}{\partial x_n}.$$
As in Case~1, movements of the readership boundary contribute no first-order profit because
$$v_n(\rho_n(x_n;x);x_n)=V_{-n}(\rho_n(x_n;x);x).$$
Applying Equation~\eqref{eqn:FOCdev} and then imposing symmetry, $x_n=x$ and $\rho_n(x;x)=h$, therefore yields
$$\left.\frac{\partial\Pi_n(x_n;x)}{\partial x_n}\right|_{x_n=x}=\frac{\lambda}{\pi}\left[\frac{\gamma h}{2\sqrt{x}}-\frac{\sin h}{2\sqrt{1-x}}\right]=0.$$
This is exactly the first-order condition in the benchmark without AI, with the unique solution
$$x^*=x_0^*=\frac{\gamma^2h^2}{\gamma^2h^2+\sin^2h}.$$
Because the two profit functions differ locally only by a constant, the strict second-order condition established in Case~1 also applies whenever $\hat{\delta}<h$, so $x^*=x_0^*$ throughout the interior of Case~2.

Case~2 ends when $\hat{\delta}$ reaches $h$, that is, when AI becomes sufficiently attractive that even the boundary consumers located at $t_i=t_n^*\pm h$ are indifferent between the content producer and the AI. Beyond this point, even setting the price to zero is no longer sufficient for content producers to retain these consumers. Therefore, the upper threshold $\overline{\alpha}(\gamma,N)$ is determined by
$$\overline{\alpha}(\gamma,N)\sqrt{2(\gamma^2+1)}=\gamma\sqrt{x^*}+\sqrt{1-x^*}\cos h.$$
Equivalently, we obtain
$$\overline{\alpha}(\gamma,N)=\frac{\gamma^2h+\sin h\cos h}{\sqrt{2(\gamma^2+1)\left(\gamma^2h^2+\sin^2h\right)}}.$$

We finally consider Case~3. Suppose that $\alpha>\overline{\alpha}(\gamma,N)$. Under Assumption~\ref{ass:alpha}, AI is sufficiently attractive to capture some consumers, but not so attractive as to fully displace content producers. Hence, there exists a unique threshold $\delta^*\in(0,h)$ such that consumers with $|t_i-t_n^*|\leq \delta^*$ choose content producer $n$, whereas consumers with $\delta^*<|t_i-t_n^*|\leq h$ choose AI.

The marginal consumer located at $|t_i-t_n^*|=\delta^*$ is indifferent between producer $n$ and AI. Therefore, $\delta^*$ satisfies
$$
\alpha\sqrt{2(\gamma^2+1)}=\gamma\sqrt{x^*}+\sqrt{1-x^*}\cos\delta^*.
$$

At the symmetric equilibrium, the target directions remain unchanged and are equally spaced around the horizontal circle. Regarding the vertical emphasis $x_n$, the AI synthesizer is now the binding outside option for every consumer that producer $n$ serves, $O_n(t; x) = v_{AI}$, and producer $n$'s profit from choosing vertical emphasis $x_n$ when all other producers choose $x$ is
$$\Pi_n(x_n; x) = \frac{\lambda}{\pi}\int_0^{\rho_n(x_n; x)}\left[\gamma\sqrt{x_n} + \sqrt{1-x_n}\cos t- \alpha\sqrt{2(\gamma^2+1)}\right] dt,$$
where the readership cutoff solves the indifference condition
$$\gamma\sqrt{x_n} + \sqrt{1-x_n}\cos\rho_n(x_n; x)
= \alpha\sqrt{2(\gamma^2+1)},$$
with $\rho_n(x^*; x^*) = \delta^*$. Differentiating $\Pi_n(x_n; x)$ with respect to $x_n$ and applying Leibniz's rule, the boundary term vanishes because the integrand is zero at $t=\rho_n(x_n; x)$. Equation~\eqref{eqn:FOCdev} therefore applies, and evaluating it at the symmetric profile yields
$$x^* = \frac{\gamma^2(\delta^*)^2}{\gamma^2(\delta^*)^2 +\sin^2\delta^*}.$$
The second-order condition is obtained as in Case~1, except that $\partial\rho_n/\partial x_n$ now follows from the indifference condition, in which the outside option does not vary with $t$. At the symmetric candidate profile,
$$\left.\frac{\partial^2 \Pi_n(x_n; x)}{\partial x_n^2}\right|_{x_n = x = x^*}= \frac{\lambda}{4\pi(\delta^*)^2(1-x^*)^{3/2}\sin\delta^*}\left[(\sin\delta^* - \delta^*\cos\delta^*)^2- \frac{(\delta^*)^2\sin^2\delta^*}{x^*}\right] < 0.$$
The inequality follows from $0 < \sin\delta^* - \delta^*\cos\delta^* 
<\delta^*\sin\delta^*$, which holds since $\delta^* < h \le \pi/2$, together with $0 < x^* < 1$. Thus, the first-order condition characterizes producer $n$'s optimal vertical dimension choice in Case~3 as well.

It remains to establish that the equilibrium pair $(x^*, \delta^*)$ exists and is unique. Substituting the first-order condition into the marginal consumer's indifference condition, using
$$\gamma\sqrt{x^*} = \frac{\gamma^2\delta^*}{\sqrt{\gamma^2(\delta^*)^2+ \sin^2\delta^*}}, \quad\sqrt{1-x^*} = \frac{\sin\delta^*}{\sqrt{\gamma^2(\delta^*)^2
+ \sin^2\delta^*}},$$
shows that any solution of the two equilibrium conditions must satisfy
$\alpha = \Phi(\gamma, \delta^*)$, where
\begin{equation}
\Phi(\gamma, z) \equiv \frac{\gamma^2 z + \sin z\cos z}{\sqrt{2(\gamma^2+1)\left(\gamma^2 z^2 + \sin^2 z\right)}}.
\label{eqn:Phi}
\end{equation}
Conversely, any $\delta^* \in (0, h]$ with $\alpha = \Phi(\gamma,\delta^*)$, together with the $x^*$ given by the first-order condition, solves both conditions. Solutions are therefore in one-to-one correspondence with roots of $\Phi(\gamma, \cdot) = \alpha$ on $(0, h]$. Direct differentiation gives
$$\frac{\partial \Phi(\gamma, z)}{\partial z}= \frac{\gamma^2 q(z) - \sin^4 z}{\sqrt{2(\gamma^2+1)}\left(\gamma^2 z^2 + \sin^2 z\right)^{3/2}},\quad \text{where}\quad q(z) \equiv z^2\cos(2z) - z\sin(2z) + \sin^2 z.$$
Since $q(z) = (\sin z - z\cos z)^2 - z^2\sin^2 z$, the inequalities $0 < \sin z - z\cos z < z\sin z$ on $(0, \pi/2]$ imply $q(z) < 0$ there; as $h = \pi/N \le \pi/2$, the function $\Phi(\gamma, \cdot)$ is continuous and strictly decreasing on $(0, h]$. Its boundary values are $\Phi(\gamma, h) = \overline{\alpha}(\gamma, N)$, by the expression for $\overline{\alpha}(\gamma, N)$ derived above, and $\lim_{z \to 0^+}\Phi(\gamma, z) = 1/\sqrt{2}$, since $\sin z = z +O(z^3)$ implies that the numerator equals $(\gamma^2+1)z + O(z^3)$ while $\gamma^2 z^2 + \sin^2 z = (\gamma^2+1)z^2 + O(z^4)$. In particular, $\overline{\alpha}(\gamma, N) < 1/\sqrt{2}$, so the Case-3 range $(\overline{\alpha}(\gamma, N), 1/\sqrt{2})$ is nonempty. For every $\alpha$ in this range, the intermediate value theorem and strict monotonicity deliver exactly one $\delta^* \in (0, h)$ with $\Phi(\gamma, \delta^*) = \alpha$, and the first-order condition then uniquely determines $x^*$. This completes the proof.

\subsection*{Proof of Proposition~\ref{prop:AI_noAI_comparison}:}
Let
$$X(\gamma,z) \equiv \frac{\gamma^2 z^2}{\gamma^2 z^2 + \sin^2 z}.$$
Differentiating with respect to $z$ gives
$$\frac{\partial X(\gamma,z)}{\partial z} = \frac{2\gamma^2 z \sin z \left(\sin z - z\cos z\right)}{\left(\gamma^2 z^2 + \sin^2 z\right)^2} > 0$$
for $z \in (0, \pi/2]$, because $\sin z - z\cos z > 0$ on this interval. Every evaluation point below lies in this range, since $h = \pi/N \le \pi/2$.

Recall from the proof of Theorem~\ref{thm:AIequilibrium} that, in Case~3,
$$\alpha = \Phi(\gamma,\delta^*), \quad x^* = X(\gamma,\delta^*),$$
where $\Phi$ is defined in Equation~\eqref{eqn:Phi}. That proof also establishes that $\partial\Phi(\gamma,z)/\partial z < 0$ for $z \in (0,h]$, and that $\overline{\alpha}(\gamma,N) = \Phi(\gamma,h)$. As shown in the proof of Theorem~\ref{thm:AIequilibrium}, whose Case~1 derivation does not involve AI and therefore applies to the benchmark economy as well, the symmetric equilibrium level of vertical emphasis in the benchmark economy is
$$x_0^* = X(\gamma,h) = \frac{\gamma^2 h^2}{\gamma^2 h^2 + \sin^2 h}.$$
We now prove each point of the proposition.

Point~1 follows directly from Theorem~\ref{thm:AIequilibrium}: in all three regimes, the equilibrium target directions are equally spaced at angular distance $2h$, as in the benchmark economy without AI.

For Point~2, if $\alpha \le \overline{\alpha}(\gamma,N)$, then Theorem~\ref{thm:AIequilibrium} gives $x^* = x_0^*$. If instead $\alpha > \overline{\alpha}(\gamma,N)$, then $\delta^* < h$. Since $X(\gamma,z)$ is strictly increasing in $z$,
$$x^* = X(\gamma,\delta^*) < X(\gamma,h) = x_0^*.$$

For Point~3, suppose that $\alpha > \overline{\alpha}(\gamma,N)$. The editorial gap is
$$x_0^* - x^* = X(\gamma,h) - X(\gamma,\delta^*).$$
Implicitly differentiating $\alpha = \Phi(\gamma,\delta^*)$ with respect to $\alpha$ gives
$$1 = \frac{\partial\Phi(\gamma,\delta^*)}{\partial z}\frac{\partial\delta^*}{\partial\alpha},$$
and hence
$$\frac{\partial\delta^*}{\partial\alpha} = \left[\frac{\partial\Phi(\gamma,\delta^*)}{\partial z}\right]^{-1} < 0.$$
Thus, a more reliable AI reduces the audience retained by each content producer. Since $X$ is strictly increasing,
$$\frac{\partial(x_0^* - x^*)}{\partial\alpha} = -\frac{\partial X(\gamma,\delta^*)}{\partial z}\frac{\partial\delta^*}{\partial\alpha} > 0.$$
Therefore, the editorial gap is strictly increasing in $\alpha$.

To study the effect of $N$, hold $\alpha$ and $\gamma$ fixed and consider $N_2 > N_1$ such that $\alpha > \overline{\alpha}(\gamma,N_j)$, $j \in \{1,2\}$. The cutoff $\delta^*$ is the unique solution to $\alpha = \Phi(\gamma,\delta^*)$, which does not depend directly on $N$. Hence, $\delta^*$, and therefore $X(\gamma,\delta^*)$, are the same in the two economies. By contrast, $h_2 = \pi/N_2 < \pi/N_1 = h_1$. Because $X(\gamma,z)$ is strictly increasing in $z$,
$$X(\gamma,h_2) - X(\gamma,\delta^*) < X(\gamma,h_1) - X(\gamma,\delta^*).$$
Thus, among economies that are both in Case~3, the editorial gap is strictly decreasing in $N$.

Finally, consider Point~4. Because $\overline{\alpha}(\gamma,N) = \Phi(\gamma,h)$ and $h = \pi/N$ is fixed when differentiating with respect to $\gamma$,
$$\frac{\partial\overline{\alpha}(\gamma,N)}{\partial\gamma} = \frac{\partial\Phi(\gamma,h)}{\partial\gamma}.$$
More generally,
$$\frac{\partial\Phi(\gamma,z)}{\partial\gamma} = \frac{\gamma\left\{\gamma^2 z\left[(z\cos z - \sin z)^2 + z^2\sin^2 z\right] + \sin z\left[2z\sin z - \cos z\left(z^2 + \sin^2 z\right)\right]\right\}}{\sqrt{2}(\gamma^2+1)^{3/2}\left(\gamma^2 z^2 + \sin^2 z\right)^{3/2}}.$$
For $z \in (0, \pi/2]$, the first term inside braces is strictly positive. Moreover,
$$2z\sin z - \cos z\left(z^2 + \sin^2 z\right)= z\left(\sin z - z\cos z\right) + \sin z\left(z - \sin z\cos z\right)> 0,$$
because $\sin z - z\cos z > 0$ and $z - \sin z\cos z > 0$ on this interval; the second inequality holds since $\sin z\cos z = \sin(2z)/2 < z$. It follows that $\partial\Phi(\gamma,z)/\partial\gamma > 0$. Evaluating at $z = h$ therefore gives
$$\frac{\partial\overline{\alpha}(\gamma,N)}{\partial\gamma} > 0.$$
Hence, $\overline{\alpha}(\gamma,N)$ is strictly increasing in $\gamma$, which completes the proof.

\subsection*{Proof of Proposition~\ref{prop:consumeralignment}:}
We first derive consumer evaluation alignment in the absence of AI and then characterize alignment when AI is adopted.

Consider two consumers $i$ and $j$ who obtain information from content producers $n$ and $m$, respectively. Under a symmetric equilibrium with vertical emphasis $x$, their posterior evaluations are
$$
z_{P,i}=\frac{\theta_i\cdot b_n}{2}s_{in}, \quad z_{P,j}=\frac{\theta_j\cdot b_m}{2}s_{jm}.
$$

For consumers in the corresponding producers' equilibrium audiences, the coefficients multiplying the signals are positive and therefore do not affect their correlation. Hence,
$$
\operatorname{Corr}(z_{P,i},z_{P,j}\mid t_i,t_j) = \operatorname{Corr}(s_{in},s_{jm}).
$$

Because the idiosyncratic signal noises are independent across consumers and independent of the underlying state,
$$
\operatorname{Cov}(s_{in},s_{jm})=b_n\cdot b_m, \quad
\operatorname{Var}(s_{in}) = \operatorname{Var}(s_{jm}) =2.
$$
It follows that
$$
\operatorname{Corr}(z_{P,i},z_{P,j}\mid t_i,t_j) = \frac{b_n\cdot b_m}{2} = \frac{x+(1-x)\cos(t_n^*-t_m^*)}{2}.
$$

Under the symmetric equilibrium, producers' target directions are equally spaced. Therefore,
$$
\mathbb E\left[\cos\left(t_n^*-t_m^*\right)\right] = \frac{1}{N^2}\sum_{n=1}^N\sum_{m=1}^N\cos(t_n^*-t_m^*) = \left|\frac{1}{N}\sum_{n=1}^N e^{it_n^*}\right|^2=0.
$$
Therefore, 
$$
\mathbb E\!\left[\operatorname{Corr}(z_{P,i},z_{P,j}\mid t_i,t_j)\right] = \frac{x}{2}.
$$

Without AI, all consumers obtain information from content producers and
$$
x_0^* = \frac{\gamma^2h^2}{\gamma^2h^2+\sin^2h}.
$$
Hence,
$$
\mathcal S_0 = \frac{x_0^*}{2} = \frac{\gamma^2h^2}{2(\gamma^2h^2+\sin^2h)}.
$$

We next derive the consumer evaluation alignment in the presence of AI. By Theorem~\ref{thm:AIequilibrium}, for
$0<\alpha\leq\overline{\alpha}(\gamma,N)$, AI affects
neither producers' content strategies nor consumers' information
acquisition choices: every consumer continues to obtain information from
content producers, and $x^*=x_0^*$. It follows immediately that
$$
\mathcal{S}_{AI} = \frac{x^*}{2} = \frac{x_0^*}{2} = \mathcal{S}_0,
$$
which establishes the first part of the proposition.

We next consider the case in which $\overline{\alpha}(\gamma,N)<\alpha<1/\sqrt{2}$. In this region, Theorem~\ref{thm:AIequilibrium} shows that AI attracts a positive share of consumers. Let
$$
\mu_P=\frac{\delta^*}{h}, \quad \mu_{AI}=1-\frac{\delta^*}{h}
$$
denote the equilibrium shares of consumers obtaining information from content producers and AI, respectively. We first derive evaluation alignment for each of the three possible pairs of information sources.

If both consumers use content producers, we have derived that 
$$
\mathbb E\!\left[\operatorname{Corr}(z_{P,i},z_{P,j}\mid t_i,t_j)\right] = \frac{x^*}{2}.
$$
For a producer user $i$ served by producer $n$ and an AI user $j$, the conditional correlation is
$$
\operatorname{Corr}(z_{P,i},z_{AI,j}\mid t_i,t_j)
=\frac{\alpha\left(\gamma\sqrt{x^*}+\sqrt{1-x^*}\cos(t_j-t_n^*)\right)}{\sqrt{2(\gamma^2+1)}}.
$$
For two AI users, Assumption~\ref{ass:informationacrossconsumers} gives
$$
\operatorname{Corr}(z_{AI,i},z_{AI,j}\mid t_i,t_j)
=\frac{\alpha^2\left(\gamma^2+\cos(t_i-t_j)\right)}{\gamma^2+1}.
$$
Conditional on the consumers' channel choices, the horizontal terms average to zero, because the audience of each channel is invariant under rotation by $2h$. The average conditional correlations for producer-AI pairs and AI-AI pairs are therefore
$$
\frac{\alpha\gamma\sqrt{x^*}}{\sqrt{2(\gamma^2+1)}}\quad\text{and}\quad\frac{\alpha^2\gamma^2}{\gamma^2+1},
$$
respectively.

Because two independently drawn consumers are both producer users with probability $\mu_P^2$, form a mixed pair of one producer user and one AI user with probability $2\mu_P\mu_{AI}$, and are both AI users with probability $\mu_{AI}^2$, we then have
$$
\mathcal{S}_{AI}
=\mu_P^2\frac{x^*}{2}
+2\mu_P\mu_{AI}\frac{\alpha\gamma\sqrt{x^*}}{\sqrt{2(\gamma^2+1)}}
+\mu_{AI}^2\frac{\alpha^2\gamma^2}{\gamma^2+1}
=\left[\mu_P\sqrt{\frac{x^*}{2}}+\mu_{AI}\frac{\alpha\gamma}{\sqrt{\gamma^2+1}}\right]^2.
$$
By Case 3(d) of Theorem~\ref{thm:AIequilibrium} and $\sin\delta^*<\delta^*$,
$$
x^*=\frac{\gamma^2(\delta^*)^2}{\gamma^2(\delta^*)^2+\sin^2\delta^*}>\frac{\gamma^2}{\gamma^2+1}.
$$
Together with Assumption~\ref{ass:alpha} and Point 2 of Proposition \ref{prop:AI_noAI_comparison}, this implies
$$\frac{\alpha\gamma}{\sqrt{\gamma^2+1}}<\frac{\gamma}{\sqrt{2(\gamma^2+1)}}<\sqrt{\frac{x^*}{2}}<\sqrt{\frac{x_0^*}{2}}=\sqrt{\mathcal{S}_0}.
$$
Since $\mu_P+\mu_{AI}=1$ and both shares are nonnegative,
$$\sqrt{\mathcal{S}_{AI}}=\mu_P\sqrt{\frac{x^*}{2}}+\mu_{AI}\frac{\alpha\gamma}{\sqrt{\gamma^2+1}}<\sqrt{\mathcal{S}_0}.
$$
Therefore, $\mathcal{S}_{AI}<\mathcal{S}_0$, which completes the proof.

\subsection*{Proof of Lemma~\ref{lm:personalizedAI}:}
Let $G_i$ denote the event that AI successfully recovers the underlying state for consumer $i$, so that $\Pr(G_i)=\alpha$. Conditional on successful synthesis, $a_i^{AI}=\omega$, and the personalized signal is $s_i^{AI} \mid G_i = \theta_i\cdot\omega/\sqrt{\gamma^2+1} + \varepsilon_i^{AI} \sim N(0,2)$. Conditional on failed synthesis, $a_i^{AI}=\zeta_i$, where $\zeta_i$ is an independent draw from the prior distribution of $\omega$. Since $\zeta_i$ has the same marginal distribution as $\omega$, $s_i^{AI}\mid G_i^c = \theta_i\cdot\zeta_i/\sqrt{\gamma^2+1} + \varepsilon_i^{AI} \sim N(0,2)$. Thus, $s_i^{AI}$ has the same distribution under successful and failed synthesis. Observing the personalized signal therefore does not reveal whether synthesis succeeded, and $\Pr(G_i\mid s_i^{AI})=\alpha$.

Conditional on successful synthesis, Gaussian updating gives
$$
\mathbb E[\theta_i\cdot\omega\mid s_i^{AI},G_i]=\frac{\sqrt{\gamma^2+1}}{2}s_i^{AI}.
$$

Under failed synthesis, $s_i^{AI}$ is independent of $\omega$, and hence
$$
\mathbb E[\theta_i\cdot\omega\mid s_i^{AI},G_i^c]=0.
$$

Thus, consumer $i$'s posterior evaluation after observing the personalized AI signal is
$$
\mathbb E[\theta_i\cdot\omega\mid s_i^{AI}] = \frac{\alpha\sqrt{\gamma^2+1}}{2}s_i^{AI}.
$$

The ex ante value of information provided by the personalized AI synthesizer is therefore
$$
v_{i,AI}^{\mathrm{pers}} = \mathbb E\left[\max\left\{\mathbb E[\theta_i\cdot\omega\mid s_i^{AI}],0\right\}\right] = \alpha\lambda\sqrt{\gamma^2+1},
$$
where $\lambda=1/(2\sqrt{\pi})$.

\subsection*{Proof of Theorem~\ref{thm:endogenousAI}:}
We first characterize how source-dependent AI reliability affects a content producer's incentive to choose vertical emphasis and then apply this characterization to each of the three equilibrium cases.

Following the main model, consumer $i$'s relevant outside option when considering content producer $n$ is $O_{in} = \max\left\{V_{i,-n},\,v_{AI}(\bar{x};\eta)\right\}$, where $v_{AI}(\bar{x};\eta) = \lambda\sqrt{2(\gamma^2+1)}\,\alpha(\bar{x};\eta)$.
As in the main model, content producer $n$ extracts the difference between the consumer's value of its content and her best outside option whenever this difference is positive. That is, $p_n^*(t_i) = \max\left\{v_{in}-O_{in},\,0\right\}$.

Let $R_n$ denote content producer $n$'s equilibrium audience, and let
$$B_n = \left\{t_i\in R_n: v_{AI}(\bar{x};\eta)\geq V_{i,-n}\right\}$$
denote the subset of its audience for whom AI is the binding outside option.

At a symmetric equilibrium, write $u=|t_i-t_n^*|$ for a reader's angular distance from producer $n$'s target. Under the maintained structure, for $t_i\in R_n$,
$$
V_{i,-n}=\lambda\left[\gamma\sqrt{x^*}+\sqrt{1-x^*}\cos(2h-u)\right],
$$
which is strictly increasing in $u$ on the readership. Hence, up to measure-zero indifferences, $B_n$ is either empty or an interval centered at $t_n^*$. If AI attracts no consumers, each producer's readership has radius $h$; an empty $B_n$ gives Case 1, and a nonempty $B_n$, of radius $\hat{\delta}\in(0,h]$, gives Case 2. If AI attracts a positive share of consumers, each producer's readership has radius $\delta^*<h$. At its boundary, $\cos\delta^*>\cos(2h-\delta^*)$, so the producer strictly outperforms the best rival. The marginal consumer's indifference therefore gives $v_{AI}=v_{in}>V_{i,-n}$. Since the best rival value increases with distance from the target, AI is the binding outside option throughout $R_n$, giving Case 3. Thus, the three cases are exhaustive within the maintained equilibrium class.

Near the equilibrium profile, producer $n$'s value of information for consumers in its readership is
$$
v_{in}=\lambda\left[\gamma\sqrt{x_n}+\sqrt{1-x_n}\cos(t_i-t_n)\right].
$$
Differentiating with respect to $x_n$ gives
$$\frac{\partial v_{in}}{\partial x_n} = \lambda\left[\frac{\gamma}{2\sqrt{x_n}} - \frac{\cos(t_i-t_n)}{2\sqrt{1-x_n}}\right].$$
We can also derive that 
$$\frac{\partial v_{AI}}{\partial x_n} = \frac{\lambda\sqrt{2(\gamma^2+1)}}{N}\alpha_x(\bar{x};\eta).$$

Differentiating producer $n$'s equilibrium profit with respect to $x_n$ therefore gives
\begin{equation}\label{eq:endgenousprofit}
    \frac{\partial\Pi_n}{\partial x_n} = \frac{1}{2\pi}\int_{R_n}\frac{\partial v_{in}}{\partial x_n}\,dt_i - \frac{1}{2\pi}\int_{B_n}\frac{\partial v_{AI}}{\partial x_n}\,dt_i = \frac{1}{2\pi}\int_{R_n}\frac{\partial v_{in}}{\partial x_n}\,dt_i - \frac{\lambda\sqrt{2(\gamma^2+1)}}{2\pi N}|B_n|\,\alpha_x(\bar{x};\eta),
\end{equation}
where $|B_n|$ denotes the Lebesgue measure of $B_n$.\footnote{The boundary terms from Leibniz's rule vanish: at a readership boundary, the producer's margin is zero, and at a boundary where the binding alternative switches, the two outside-option values coincide.} The first term captures the direct effect of vertical emphasis on consumers' value for content producer $n$'s content. The second term is new relative to the main model: increasing $x_n$ improves the original content provided by content producers, raises AI reliability, and thereby strengthens the outside option of consumers for whom AI is the binding alternative.

Suppose the AI synthesizer does not bind any consumer's outside option. Then $B_n=\varnothing$ for every content producer. The second term in Equation (\ref{eq:endgenousprofit}) therefore vanishes, and the producer’s editorial first-order condition at this equilibrium is identical to that in the benchmark without AI. Hence, $x^*=x_0^*$. Because AI does not bind consumers' outside options, equilibrium prices and information acquisition choices also coincide with those in the no-AI benchmark. This establishes part 1.

Suppose all consumers continue to obtain information from content producers, but AI is the binding outside option for consumers within distance $\hat{\delta}\in(0,h]$ of each producer's target direction. The consumer at the boundary between the region in which AI binds and the region in which the adjacent producer binds is indifferent between these two outside options. Under symmetry, this consumer is at distance $2h-\hat{\delta}$ from the adjacent producer's target. Equating the value of the adjacent producer with the value of AI gives
$$\sqrt{2(\gamma^2+1)}\,\alpha(x^*;\eta) = \gamma\sqrt{x^*} + \sqrt{1-x^*}\cos(2h-\hat{\delta}).$$

Each content producer continues to serve an audience of radius $h$, so $R_n=[t_n^*-h,t_n^*+h]$. Moreover, AI is the binding outside option over an interval of total length $|B_n|=2\hat{\delta}$. Thus,
\begin{equation}\label{eq:case2endogenous}
    \int_{R_n}\frac{\partial v_{in}}{\partial x_n}\,dt_i = \lambda\int_{-h}^{h}\left[\frac{\gamma}{2\sqrt{x^*}} - \frac{\cos u}{2\sqrt{1-x^*}} \right]du\\ = \lambda \left[\frac{\gamma h}{\sqrt{x^*}} - \frac{\sin h}{\sqrt{1-x^*}}\right].
\end{equation}

Substituting Equation (\ref{eq:case2endogenous}) and $|B_n|=2\hat{\delta}$ into the first-order condition from Equation (\ref{eq:endgenousprofit}), and canceling common terms, gives
\begin{equation}\label{eq:case2condition2endo}
    \frac{\gamma h}{2\sqrt{x^*}} - \frac{\sin h}{2\sqrt{1-x^*}} = \frac{\sqrt{2(\gamma^2+1)}\,\hat{\delta}}{N}\alpha_x(x^*;\eta).
\end{equation}

It remains to establish the comparison with $x_0^*$. Define
$$D_h(x) \equiv \frac{\gamma h}{2\sqrt{x}} - \frac{\sin h}{2\sqrt{1-x}}.$$
The no-AI equilibrium satisfies $D_h(x_0^*)=0$.
Moreover,
$$\frac{\partial D_h(x)}{\partial x} = -\frac{\gamma h}{4x^{3/2}} - \frac{\sin h}{4(1-x)^{3/2}} <0,$$
so $D_h(x)$ is strictly decreasing in $x$.

By Assumption~\ref{ass:sourcereliability}, the right-hand side of Equation~\eqref{eq:case2condition2endo} is nonnegative. Hence, $D_h(x^*)\geq0$. Since $D_h$ is strictly decreasing and $D_h(x_0^*)=0$, it follows that $x^*\leq x_0^*$. If $\alpha_x(x^*;\eta)>0$, the right-hand side of Equation~\eqref{eq:case2condition2endo} is strictly positive. Therefore, $D_h(x^*)>0$, which, by strict monotonicity of $D_h$, implies $x^*<x_0^*$.
This establishes part 2.

Suppose now that the AI synthesizer attracts a positive share of consumers. Let $\delta^*\in(0,h)$ denote the radius of each producer's retained audience. The marginal consumer at distance $\delta^*$ from the producer's target is indifferent between obtaining information from the producer and from AI. Therefore,
$$\gamma\sqrt{x^*} + \sqrt{1-x^*}\cos\delta^* = \sqrt{2(\gamma^2+1)}\alpha(x^*;\eta).$$

Because AI is the binding outside option for all consumers retained by the producer, $R_n=B_n = [t_n^*-\delta^*,t_n^*+\delta^*]$, and hence $|R_n|=|B_n|=2\delta^*$. Therefore,
$$\int_{R_n}\frac{\partial v_{in}}{\partial x_n}\,dt_i = \lambda\int_{-\delta^*}^{\delta^*}\left[\frac{\gamma}{2\sqrt{x^*}} - \frac{\cos u}{2\sqrt{1-x^*}} \right]du = \lambda\left[\frac{\gamma\delta^*}{\sqrt{x^*}} -
\frac{\sin\delta^*}{\sqrt{1-x^*}}\right].$$
Substituting this equation and $|B_n|=2\delta^*$ into Equation~\eqref{eq:endgenousprofit} and imposing the first-order condition yields
$$\frac{\gamma\delta^*}{2\sqrt{x^*}} - \frac{\sin\delta^*}{2\sqrt{1-x^*}} =\frac{\sqrt{2(\gamma^2+1)}\,\delta^*}{N}\alpha_x(x^*;\eta),$$
which is the second equilibrium condition in part 3.

To establish $x^*<x_0^*$, for any $z\in(0,h]$ define
$$D_z(x) \equiv \frac{\gamma z}{2\sqrt{x}} - \frac{\sin z}{2\sqrt{1-x}}.$$
As above, $\partial D_z(x)/\partial x <0$. The unique solution to $D_z(x)=0$ is $X(\gamma,z) \equiv (\gamma^2z^2)/(\gamma^2z^2+\sin^2z)$. Similar to the logic in part 2, $D_{\delta^*}(x^*)\geq0$. Because $D_{\delta^*}$ is strictly decreasing and $D_{\delta^*}\bigl(X(\gamma,\delta^*)\bigr)=0$, we obtain $x^* \leq X(\gamma,\delta^*)$.

Next, $X(\gamma,z)$ is strictly increasing in $z$ for $z\in(0,h]$. Since AI adoption implies $\delta^*<h$, we have $X(\gamma,\delta^*) < X(\gamma,h) = x_0^*$. We thus obtain $x^*<x_0^*$. This establishes part 3 and completes the proof.

\end{document}